\documentclass[aps,nofootinbib]{revtex4}%
\usepackage{amsfonts}
\usepackage{amsmath}
\usepackage{amssymb}
\usepackage{bbold}
\usepackage{graphicx}
\usepackage{array}
\usepackage{hhline}
\usepackage{calc}
\usepackage{color}
\usepackage{comment}

\providecommand{\U}[1]{\protect\rule{.1in}{.1in}}
\newtheorem{theorem}{Theorem}

\newtheorem{proposition}[theorem]{Proposition}

\newenvironment{proof}[1][Proof]{\noindent\textbf{#1.} }{\ \rule{0.5em}{0.5em}}

\begin{document}

\title{Quasi-quantization: classical statistical theories with an epistemic restriction} 

\author{Robert W. Spekkens}
\affiliation{Perimeter Institute for Theoretical Physics, 31 Caroline St. N, Waterloo, \\
Ontario, Canada N2L 2Y5}

\date{Sept. 16, 2014}

\begin{abstract}
A significant part of quantum theory can be obtained from a single innovation relative to classical theories, namely, that there is a fundamental restriction on the sorts of statistical distributions over physical states that can be prepared.  This is termed an ``epistemic restriction'' 
because it implies a fundamental limit on the amount of knowledge that any observer can have about the physical state of a classical system.  This article provides an overview of {\em epistricted theories}, that is, theories that start from a classical statistical theory and apply an epistemic restriction.  We consider both continuous and discrete degrees of freedom, and
 show that a particular epistemic restriction called {\em classical complementarity}
provides the beginning of a unification of all known epistricted theories. 
This restriction appeals to the symplectic structure of the underlying classical theory and consequently can be applied to an arbitrary classical degree of freedom.
As such, it can be considered as a kind of \emph{quasi-quantization} scheme; ``quasi'' because 
it generally only yields a theory describing a subset of the preparations, transformations and measurements allowed in the full quantum theory for that degree of freedom, and because in some cases, such as for binary variables, it 
yields a theory that is a distortion of such a subset.
Finally, we propose to classify quantum phenomena as weakly or strongly nonclassical by whether or not they can arise in an epistricted theory.
\end{abstract}

\maketitle
\tableofcontents

\section{Introduction}

\subsection{Epistricted theories}



Start with a classical theory for some degree of freedom and consider the statistical theory associated with
it.  This is the theory that describes the statistical distributions over the space of physical states and how they change over time.  If one then postulates, as a fundamental principle, that there is a restriction on what kinds of
statistical distributions can be prepared, then the resulting theory reproduces a large part of
quantum theory, in the sense of reproducing precisely its operational predictions.  This article reviews recent work on such theories and their relevance for notions of nonclassicality, for the interpretation of the quantum state, and for the program of deriving the formalism of quantum theory from axioms. 


Some clarifications are in order regarding statistical theories.  Given a system whose physical state is drawn from some ensemble of possibilities, the statistical distribution associated to this ensemble can be taken either to describe relative frequencies of physical properties within the virtual ensemble or it can be taken to describe the knowledge that an agent has about an individual system when she knows that it was drawn from that ensemble.  The latter sort of language is preferred by those who take a Bayesian approach to statistics, and we shall adopt it here. 
The distinction between the physical state of a system and an agent's state of knowledge of that physical state will be critical in what follows.  As such, we will make use of some jargon to clearly distinguish the two sorts of states. 
Recalling the Greek terms for reality and for knowledge, \emph{ontos} and \emph{epist\={e}m\={e}}, we will henceforth refer to  physical states as \emph{ontic states} and states of knowledge as \emph{epistemic states}~\cite{spekkens2007evidence}.  The theory that governs the evolution of ontic states is an \emph{ontological theory}, while the {\em statistical theory} describes the evolution of epistemic states. A restriction on knowledge is an \emph{epistemic restriction}.  The theories we are considering, therefore, are epistemically-restricted statistical theories of classical systems.   Given that this is a rather unwieldy descriptor, we introduce the term \emph{epistricted theory} as an abbreviation to it.


It is worth considering in a bit more detail the scheme by which one infers from a given classical theory the epistricted version thereof. 
One starts with a particular classical ontological theory (first column of Table~\ref{Table:epistrictedtheories}).  We are here considering the usual notion of an ontological theory:
 one which provides a kinematics and a dynamics, that is, a hypothesis about the possible physical states that a system can occupy at a given time and a law describing how each state can evolve over time.  

One then constructs the statistical theory for the classical ontological theory under consideration (second column of Table~\ref{Table:epistrictedtheories}).  
The fundamental object here is a statistical distribution over the physical state space rather than a point in the physical state space, that is, an epistemic state rather than an ontic state. The statistical theory answers questions such as: If the physical state of the system undergoes deterministic dynamics, how does the statistical distribution change over time? Or, more precisely, if an agent assigns a statistical distribution over physical states at one time and she knows the dynamics, what statistical distribution should she assign at a later time?
If an agent implements a measurement on the system and takes note of the outcome, how should she update her statistical distribution?  

In the final and most significant step of the theory-construction scheme, 
 one postulates a fundamental restriction on the sorts of statistical distributions that can describe an agent's knowledge of the system (third column of Table~\ref{Table:epistrictedtheories}).   

As a first example, consider the classical ontological theory of particle mechanics.  The associated statistical theory is what is sometimes called \emph{Liouville mechanics}.
If one then postulates a classical version of the uncertainty principle as the epistemic restriction~\cite{bartlett2012reconstruction}, 
then one obtains a theory which we shall refer to here as \emph{Gaussian epistricted mechanics} (it was called  \emph{epistemically-restricted Liouville mechanics} in Ref.~\cite{bartlett2012reconstruction}).  This theory is equivalent to a subtheory of quantum mechanics, the Gaussian subtheory, which is defined in Ref.~\cite{bartlett2012reconstruction}.

The case of optics is a straightforward extension of the case of mechanics because each optical mode is a scalar field and the phase spaces of a field mode and of a particle are both Euclidean.
The canonically conjugate variables, which are position and momentum in the mechanical case,  are field quadratures in the optical case.  The statistical theory of optics is well-known~\cite{born1999principles}.  Upon postulating an epistemic restriction in the form of an uncertainty principle, one obtains the optical analogue of the Gaussian subtheory of quantum mechanics, namely, the Gaussian subtheory of quantum optics, which is sometimes referred to as \emph{linear quantum optics}. The latter theory includes a wide variety of quantum optical experiments.


\begin{table*}[htb]
  \begin{tabular}{|c|c|c|}
    \hline
    & {\bf Statistical theory} & {\bf Epistemically-restricted statistical theory }\\
    {\bf Classical ontological theory} & {\bf for the classical ontological theory} & {\bf for the classical ontological theory} 
         \tabularnewline
\hline
Mechanics & Liouville mechanics & Gaussian epistricted mechanics \\
&& \color{red} =Gaussian subtheory of quantum mechanics \\
&& Quadrature epistricted mechanics\\ 
&&  \color{red} =Quadrature subtheory of quantum mechanics \\
\hline
trits & Statistical Theory of trits & Quadrature epistricted theory of trits \\
&& \color{red} = Quadrature/Stabilizer subtheory for qutrits \\
\hline
Bits & Statistical Theory of bits & Quadrature epistricted theory of bits \\
&& \color{red} $\simeq$ Quadrature/Stabilizer subtheory for qubits \\
\hline
Optics & Statistical optics & Gaussian epistricted optics \\
&& \color{red} = Gaussian subtheory of quantum optics \\
& & Qudrature epistricted optics \\
&& \color{red} = quadrature subtheory of quantum optics \\
\hline    
  \end{tabular}
 \caption{Theories arising from imposing certain epistemic restrictions on statistical theories for various classical theories, and the subtheories of quantum theory that they correspond to. } 
 \label{Table:epistrictedtheories}
\end{table*}

One can apply the same strategy for a classical ontological theory wherein the fundamental degrees of freedom are discrete, so that every system has an integer number $d$ of ontic states.  It is unusual for physicists to discuss discrete degrees of freedom in a classical context.  Nonetheless, this is done when considering the possibility of models that are cellular automata.  It is also common when describing the physics of digital computers.  The language of computation, therefore, is a natural one for describing such a theory.  

The simplest case to consider is $d=2$, in which case the fundamental degree of freedom is a bit.  
 A collection of such fundamental degrees of freedom corresponds to a string of bits.  An interaction between two distinct degrees of freedom can be understood as a gate acting on two bits.  Similarly for interactions between $n$ systems.
 General dynamics, which corresponds to an arbitrary sequence of interactions, can be understood as a circuit.
The statistical theory of bits is just a theory of the statistical distributions over the possible bit-strings, how these evolve under gates, and how these are updated as a result of registering the outcome of measurements performed on the bits.  One then imposes a restriction on what kinds of statistical distributions can characterize an agent's knowledge of the value of the bit-string.  

This is the arena in which the first epistricted theory was constructed~\cite{spekkens2007evidence}.  
The restriction on knowledge was implemented through a principle that asserted that any agent could have the answers to at most half of a set of questions that would specify the ontic state of the system.  Consequently, when one has maximal knowledge, then the number of independent questions that are answered is equal to the number of independent questions that are unanswered; in this case, one's measure of knowledge is equal to one's measure of ignorance. 
This epistemic restriction was dubbed the \emph{knowledge-balance principle}, and the epistricted theory of bits that resulted was called a \emph{toy theory} in Ref.~\cite{spekkens2007evidence}.
This theory mirrors very closely a subtheory of the quantum theory of qubits, namely, the one which is known to quantum information theorists as the \emph{stabilizer formalism} and which we will term the \emph{stabilizer subtheory} of qubits.  It will be presented in Sec.~\ref{quadraturesubtheories}.  Stabilizer states are defined to be the eigenstates of products of Pauli operators, stabilizer measurements are measurements of commuting sets of products of Pauli operators, and stabilizer transformations are unitary transformations which take stabilizer states to stabilizer states.    Although the toy theory is not operationally equivalent to the stabilizer subtheory, it reproduces qualitatively the same phenomenology.  Also, the toy theory can be cast in the same sort of language as the stabilizer theory, as noted in Ref.~\cite{pusey2012stabilizer}.

Subsequent work sought to develop an epistricted theory for discrete systems with $d$ ontic states, where $d > 2$.   There were two natural avenues to pursue: generalize the knowledge-balance principle used in Ref.~\cite{spekkens2007evidence} or devise a discrete version of the classical uncertainty principle used in Ref.~\cite{bartlett2012reconstruction}.  
The former approach was pursued by van Enk~\cite{van2007toy}.\footnote{We are here refering to the first part of Ref.~\cite{van2007toy}.  In the second part, the author proposes a theory wherein there is a restriction on what can be known about the outcome of measurements, rather than a restriction on what can be known about some underlying ontic state.  As such, the latter theory is not an epistricted theory.}
However, some important work by Gross~\cite{gross2006hudson} established that it is possible to define a discrete phase space for a $d$-level systems where $d$ is an odd prime and it is possible to define a Wigner representation based on this phase space such that the stabilizer theory for these qudits admits of a nonnegative Wigner representation.  Gross's Wigner representation can be understood as a hidden variable model for the stabilizer subtheory.  This suggests that one should be able to define a classical theory of $d$-level systems using this discrete phase space and then to find an epistemic restriction that yields precisely this hidden variable model.  
In other words, Gross's work strongly suggests that one should look for an epistemic restriction that appeals to the phase-space structure, analogously to the epistemic restriction that was used in the Gaussian epistricted mechanics~\cite{bartlett2012reconstruction}.

Such an epistemic restriction was subsequently identified~\cite{Sch08}.  Using the phase-space structure, one can define \emph{quadrature variables} for the classical system.  The epistemic restriction then asserts that one can have joint knowledge of a set of quadrature variables if and only if they commute relative to a discrete analogue of the Poisson bracket.  The epistemic restriction is dubbed \emph{classical complementarity} and the theory that results is called the \emph{quadrature epistricted theory} of $d$-level classical systems. 

If we apply the complementarity-based epistemic restriction in the case of $d=2$, the resulting theory---the qudrature epistricted theory of bits---turns out to be equivalent to the toy theory of Ref.~\cite{spekkens2007evidence},
and as mentioned previously, this is operationally a close phenomenological cousin of the stabilizer theory of qubits.

On the other hand,  for $d$ an \emph{odd} prime, i.e., any prime besides 2, the quadrature epistricted theory reproduces \emph{precisely} the stabilizer theory for qudits.  For such values of $d$, the epistemic restriction of classical complementarity 
turns out to be inequivalent to the knowledge-balance principle. The latter specifies only that at most half of the full set of variables can be known, whereas the former picks out particular halves of the full set of variables, namely, the halves wherein all the variables Poisson-commute.  Because the restriction of classical complementarity actually reproduces the stabilizer theory for qudits while the knowledge-balance principle does not~\cite{Sch08}, epistemic restrictions based on the symplectic structure seem to be preferable to those based on a principle of knowledge balance. 

We will also show that on the quantum side, one can define the notion of a quadrature \emph{observable}, a quantum analogue of a classical quadrature variable.  In $d=2$, the Pauli operators are both unitary and Hermitian; as unitaries, they constitute the quantum analogue of classical phase-space displacements, while as observables, they correspond to our quadrature observables.   In $d>2$, on the other hand, the generalized Pauli operators are unitary but not always Hermitian and therefore cannot always be interpreted as observables. Consequently, the stabilizer of a state in $d>2$ specifies the unitaries that leave the state invariant, not the observables for which the state is an eigenstate.  In $d>2$, the quadrature observables are the ones that are defined in terms of the eigenbases of the generalized Pauli operators.  They provide a means for acheiving a characterization of stabilizer states for any $d$ as joint eigenstates of a commuting set of quadrature observables.  This characterization is more analogous to our characterization, in epistricted theories, of the valid epistemic states as states wherein one has joint knowledge of a Poisson-commuting set of quadrature variables.


Finally, the epistemic restriction of classical complementarity can also be applied to particle mechanics, where it is different from the restriction based on the uncertainty principle that is used in Ref.~\cite{bartlett2012reconstruction}.  In particular, a smaller set of statistical distributions are considered valid epistemic states.  Using the principle of classical complementarity, one obtains a different theory at the end, which we call \emph{quadrature epistricted mechanics}.  We prove that this is equivalent to a subtheory of quantum mechanics that we will call the  \emph{quadrature subtheory of quantum mechanics} and which we will describe in detail in Sec.~\ref{quadraturesubtheories}.  The latter stands to the Gaussian subtheory of quantum mechanics as the quadrature epistricted theory of mechanics stands to the Gaussian epistricted theory of mechanics.  One can similarly define analogous theories for optics.

It follows that the epistemic restriction of classical complementarity provides the beginning of a unification of all known epistricted theories. It can be applied for both continuous and discrete degrees of freedom, and the formalism can be made to look precisely the same in each case. 

It remains an open question whether one can find a form of the epistemic restriction that is applicable to an arbitrary degree of freedom and that when applied in the case of a $d$-level system yields the Stabilizer/quadrature subtheory of qudits while when applied in the case of continuous variable systems yields the Gaussian subtheory of quantum mechanics/optics rather than merely the quadrature subtheory. 

Guided by the bridge between the epistricted theories and the quantum subtheories, we present the formalism of the associated quantum subtheories in a unified manner for continuous and discrete degrees of freedom.  This presentation focusses on quadrature observables rather than stabilizer groups and helps to reveal the analogies between the subtheories for the different degrees of freedom.

For any epistemic restriction that is applicable to many different degrees of freedom, such as the principle of classical complementarity described here, one can think of the process of applying this restriction to the corresponding classical statistical theories as a kind of quantization scheme, or more precisely, a \emph{quasi-quantization} scheme. It is ``quasi'' because it does not succeed at obtaining the full quantum theory from its classical counterpart and because in certain cases, such as binary variables, it does not even yield a subtheory of quantum theory.\footnote{Note that for the purposes of this article, the term \textquotedblleft quantum theory\textquotedblright\ refers to a theory schema that can be applied to many different degrees of freedom: particles, fields and discrete systems.}  
Unlike normal quantization schemes, which are mathematically inspired, the quasi-quantization scheme of this approach is {\em conceptually} inspired.  There is no ambiguity about how to interpret the formalism that results.

Although our quasi-quantization scheme has already been applied to a few different sorts of degrees of freedom, it is clear that one could apply it to others.
Vector fields are a good example, one which promises the possibility of a quasi-quantization of classical electrodynamics.
By finding the appropriate epistemic restriction on a statistical theory of electrodynamics, one can imagine deriving a theory that might be equivalent to---or perhaps, as for the case of bits, merely analogous to---some subtheory of quantum electrodynamics\footnote{Note that the theory of \emph{stochastic electrodynamics} has some significant similarities to an epistricted theory of electrodynamics, but there are also significant differences. Many authors who describe themselves as working on stochastic electrodynamics posit a \emph{nondeterministic} dynamical law for the fields, whereas an epistricted theory of electrodynamics is one wherein agents merely lack knowledge of the electrodynamic fields, which continue to evolve deterministically. That being said, Boyer's version of stochastic electrodynamics~\cite{Boy80} does not posit any modification of the dynamical law and so is closer to what we are imagining here.  A second difference is that in stochastic electrodynamics, there is no epistemic restriction on the matter degrees of freedom.  However, if one degree of freedom can interact with another, 
then to enforce an epistemic restriction on one, it is necessary to enforce a similar epistemic restriction on the other.  In other words, the assumptions of stochastic electrodynamics were inconsistent. The sort of epistricted theory of electrodynamics we propose here is one that would apply the epistemic restriction to the matter and to the fields.}.  At present, it is not obvious how to do this because the epistemic restrictions that have worked best for the degrees of freedom considered thus far have made reference to canonically conjugate degrees of freedom.
One therefore expects to encounter precisely the same difficulties that were faced by those who attempted a canonical quantization of classical electrodynamics.  Presumably, therefore, it would be useful to develop a Lagrangian, or least-action quasi-quantization scheme in addition to the canonical one.
If one could succeed at devising an epistricted theory of electrodynamics, then it would of course be very interesting to attempt to apply quasi-quantization to classical theories of gravity.  This would not yield a full quantum theory of gravity, but it might reconstruct some subtheory, or a distorted version of such a subtheory.

The rest of the introduction makes explicit what can and cannot be explained in epistricted theories,
together with their significance for interpretation and axiomatization.  
We have put this material up front rather than at the end of the paper for the benefit of those readers 
who are reluctant to engage with the detailed development until they have had certain questions answered, in particular, 
questions about the precise explanatory scope of these epistricted theories, and the question of why one should care about a quantization scheme that does not recover the full quantum theory.


\subsection{Explanatory scope}


We return now to the claim that epistricted theories reproduce a ``large part'' of quantum theory.  
At this stage, a sceptic might be unconvinced 
on the grounds that for each classical ontological theory, the subtheory of the corresponding quantum theory that has been derived via this quantization scheme is \emph{far} from 
the full quantum theory.  For instance, Gaussian epistricted mechanics yields a part of quantum mechanics wherein the dynamics include
only those Hamiltonians that are at most quadratic in position and momentum observables~\cite{bartlett2012reconstruction}.  Clearly, this is a small subset of all possible Hamiltonians.  Nonetheless, we argue that the relative size of the space of Hamiltonians is not the correct metric by which to assess this project.  The primary object of the exercise is to 
achieve conceptual clarity on the principles that might underly quantum theory.  As such, it is better to ask: how many distinctively quantum phenomena are reproduced within these subtheories? In particular, how many of the phenomena that are usually taken to defy classical explanation?   In terms of the phenomena they include, 
the subtheories of quantum theory one obtains by an epistemic restriction 
\emph{do} subsume a large part of the full theory.  
In support of this claim, Table~\ref{tbl:categorization} provides a categorization of some prominent quantum phenomena into those that arise in epistricted theories (on the left), and those that do not (on the right).  
As one can easily see, for this particular list, the lion's share are found on the left, and this set includes many of the phenomena that are typically taken to provide the greatest challenge to the classical worldview.\footnote{It should be noted that many researchers had previously recognized the possibility of recovering many of these quantum phenomena if one compared quantum states to probability distributions in a classical statistical theory~\cite{caves9601025quantum,emerson,hardy1999disentangling, kirkpatrick2003quantal}}

\begin{table*}[htb]
  \begin{tabular}{|c|c|}
    \hline
    {\bf Phenomena arising in epistricted theories} & {\bf Phenomena not arising in epistricted theories} \tabularnewline
    \hline
Noncommutativity & Bell inequality violations  \tabularnewline
Coherent superposition & Noncontextuality inequality violations \tabularnewline
Collapse & Computational speed-up (if it exists) \tabularnewline
Complementarity & Certain aspects of items on the left \tabularnewline
No-cloning &  \tabularnewline
No-broadcasting & \tabularnewline
Interference & \\
Teleportation & \\
Remote steering & \\
Key distribution &\\
Dense coding &\\
Entanglement &\\
Monogamy of entanglement &\\
Choi-Jamiolkowski isomorphism &\\
Naimark extension &\\
Stinespring dilation &\\
Ambiguity of mixtures &\\
Locally immeasurable product bases &\\
Unextendible product bases &\\
Pre and post-selection effects &\\
Quantum eraser &\\
And many others... &\\
\hline    
  \end{tabular}
 \caption{Categorization of quantum phenomena.  }
 \label{tbl:categorization}
\end{table*}

Note that it is typically the case that if one looks hard enough at a given quantum phenomenon that appears on the left list, one can usually find \emph{some} feature of it that cannot be explained within an epistricted theory.
When we place a given phenomenon on the left, therefore, what we are claiming is that an epistricted theory
can reproduce \emph{the features of this phenomenon that are most frequently cited as making it difficult to understand classically}.  Consider the example of quantum teleportation.  What is most frequently taken to be mysterious about teleportation from a classical perspective is that the amount of information that is required to describe the quantum state exceeds the amount of information that is communicated in the protocol.  This is just as true, however, if one seeks to teleport a quantum state within the stabilizer theory of qubits: for a single qubit, this subtheory includes only \emph{six} distinct quantum states (rather than an infinite number), but the teleportation protocol still succeeds while communicating only \emph{two} bits of classical information, which is less than $\log_2{6}$ and hence not enough to describe a state drawn from this set.  As such, we judge the teleportation protocol in the stabilizer formalism to include the essential mystery of teleportation, and, because an epistricted theory can reproduce this notion of teleportation, we put teleportation on the left-hand list.  One can always point to features of the teleportation protocol in the full quantum theory that \emph{do not} arise in the stabilizer theory of qubits, for instance, the fact that it works for an infinite number of quantum states rather than just six.  However, this is not the feature of teleportation that is typically cited as ``the mystery''.  Such incidental features are the sorts of things that we mean to include on the right-hand side of our classification under ``certain aspects of items on the left''.  Of course, one of the lessons of this categorization exercise is that the usual story about what is mysterious about a given quantum phenomenon should be supplanted by one that highlights these more subtle features, and these should henceforth be our focus when puzzling about quantum theory.


The left-hand list includes basic quantum phenomena, such as non-commutativity of
observables, interference, coherent superposition, collapse of the wave
function, the existence of complementary bases, and a no-cloning theorem~\cite{wootters1982single}.
It also includes many quantum information processing tasks, such as teleportation~\cite{bennett1993teleporting} and key distribution~\cite{bennett1984quantum}.
A large part of entanglement theory~\cite{horodecki2009quantum} is there, as are more
exotic phenomena, such as locally indistinguishable product bases~\cite{bennett1999quantum} and unextendible product bases~\cite{bennett1999unextendible}. \ One also gets
many of the distinctive relations that hold between (and within) the sets of quantum states, quantum measurements, and quantum transformations, such as the Choi-Jamio\l{}kowski isomorphism
between bipartite states and unipartite operations~\cite{Choi1975,Jamiolkowski1972}, 
the Naimark extension of positive-operator valued measures into projector-valued measures~\cite{Naimark1940}, the Stinespring dilation of irreversible operations into reversible (unitary) operations~\cite{stinespring1955positive}, and the fact that
there are many convex decompositions and many purifications of a mixed state. 

On the right-hand side, we find Bell-inequality violations~\cite{Bell1964} and noncontextuality-inequality violations~\cite{kochen1967problem,spekkens2005contextuality,liang2011specker}.  This is expected, as these phenomena are the operational signatures of the impossiblity of locally causal and noncontextual ontological models respectively, whereas the particular sorts of epistemic restrictions that have been considered to date yield theories that satisfy local causality and noncontextuality (even the generalized sense of noncontextuality of Ref.~\cite{spekkens2005contextuality}). The right-hand side also includes quantum computational speedup,
 with the caveat that the claim of an exponential speed-up is predicated on certain unproven conjectures, such as the factoring problem being outside the complexity class P. 


There are also some phenomena that have not yet been conclusively categorized.  Two examples 
are: the quantization of quantities such as energy and angular
momentum and the statistics of indistinguishable particles. 

There is always some satisfaction in adding a quantum phenomenon to the left-hand list: it suggests that the idea of an epistemic restriction captures much of the innovation of quantum theory, that the phenomena in question is not so mysterious after all.  However, the right-hand list is the one that we would most like to see grow, because
the phenomena that appear there are the ones that still seem surprising, and it is by focusing on these that one can best develop the research program wherein quantum states are understood as states of knowledge.


Our quasi-quantization scheme sheds light on the old question ``what is the conceptual innovation of quantum theory relative to classical theories?''
In particular, 
it implies that the frontier between what \emph{can} and what \emph{cannot} be explained classically extends much deeper into  quantum territory than previously thought.  This is because in order to pronounce a phenomenon nonclassical, one should be maximally permissive in \emph{how} a classical theory manages to reproduce the phenomenon.  For instance, when considering whether certain operational statistics admit of a local or a noncontextual hidden variable model, one must allow an arbitrary space of ontic states, i.e., arbitrary hidden variables.  With the benefit of hindsight, one sees that previous assessments of the scope of classical explanations were 
overly pessimistic because
they did not consider the 
possibility that some phenomenon exhibited by quantum states was reproduced by classical \emph{epistemic states} rather than by classical \emph{ontic states.}

Phenomena arising in an epistricted theory might still be considered to exhibit a type of nonclassicality insofar as an epistemic restriction is, strictly speaking, an assumption that goes beyond classical physics.  
 But it is a weak type of nonclassicality, as it ultimately can be understood via a relatively modest addendum to the classical worldview.   By contrast, quantum phenomena that do not arise in epistricted theories constitute a strong type of nonclassicality, one which marks a significant departure from the classical worldview.  Table~\ref{tbl:categorization} may therefore be understood as sorting quantum phenomena into categories of weak and strong nonclassicality.
 
 \subsection{Interpretational significance}

Epistricted theories serve to highlight the existence (and the appeal) of a type of ontological model that has previously received almost no attention.   With the exception of a model of a qubit proposed by Kochen and Specker in 1967, previous ontological models have been such that the ontic state included a description of the quantum state, and therefore any two distinct pure quantum states necessarily described different ontic states.  This was true whether the model considered the space of ontic states to be precisely the space of pure quantum states, or whether the quantum state was supplemented by additional variables, such as occurs, for instance, in Bohmian mechanics.
 By contrast, in epistricted theories, two distinct pure quantum states that are nonorthogonal correspond to two probability distributions that overlap on one or more ontic states.  Indeed, it was the work on epistricted theories that led to the articulation of the  distinction between a $\psi$-ontic model, wherein the ontic state encodes the quantum state, and a $\psi$-epistemic model, wherein it does not~\cite{spekkens2007evidence,spekkens2005contextuality,Harrigan2010}.

For the subtheories of quantum theory described above (Gaussian and quadrature quantum mechanics and the stabilizer theory of qudits for $d$ an odd prime), $\psi$-epistemic models exist and provide a compelling causal explanation of the operational predictions of those theories.   The breadth of quantum phenomenology that is reproduced within epistricted theories suggests that something about the principles underlying these theories must be correct.  The assumption that quantum states should be interpreted as epistemic states rather than ontic states seems a good candidate. 

This in no way implies, however, that the innovation of quantum theory is \emph{merely} to impose an epistemic restriction on some underlying classical physics.  This is clearly \emph{not} the only innovation of quantum theory.  As we have noted, epistricted theories, considered as ontological models, are by construction  both local and noncontextual, and because of Bell's theorem and the Kochen-Specker theorem, we know that the full quantum theory cannot be explained by such models. If quantum computers really do allow an exponential speed-up over their classical counterparts, then this too cannot be reproduced by such models.  What the success of epistricted theories suggests, rather, is that pure quantum states and statistical distributions over classical ontic states are \emph{the same category of thing}, namely, an epistemic thing.  This is a point of view that has been central also to the ``QBist'' research program~\cite{caves2002quantum,Fuchs2003a,fuchs2013quantum}.


A question that naturally arises is whether one can construct a $\psi$-epistemic model of the \emph{full} quantum theory, or whether one can find natural assumptions under which such models are ruled out, This question was first posed by Lucien Hardy and was formalized in Ref.~\cite{Harrigan2010}.  It has become the subject of much debate in recent years~\cite{pusey2012reality,lewis2012distinct,colbeck2012system}.  It is worth noting, however, that the standard framework for such models contains many implicit assumptions,  including the idea
that the correct formalism for describing epistemic states is classical probability theory.  
This assumption can be questioned, and indeed, the nonlocality and contextuality of quantum theory already suggest that it should be abandoned, as argued in \cite{leifer2013towards,wood2012lesson}.  



The investigation of epistricted theories, therefore, need not---and indeed \emph{should not}---be considered as the first step in a research program that seeks to find a $\psi$-epistemic model of the full quantum theory.  
Even though such a model could always circumvent any no-go theorems by violating their assumptions, it would be just as unsatisfying as a $\psi$-ontic model insofar as it would need to be explicitly nonlocal and contextual.
Rather, the investigation of epistricted theories is best considered as a first step in a larger research program wherein the framework of ontological models---in particular the use of classical probability theory for describing an agent's incomplete knowledge---is ultimately rejected, but
where one holds fast to the notion that a quantum state is epistemic.
\subsection{Significance for the axiomatic program}

Most reconstruction efforts are focussed on recovering the formalism of the full quantum theory.
However, it may be that  there are particularly elegant axiomatic schemes that are not currently in our reach
and that the road to progress involves temporarily setting one's sights a bit lower. 
The quasi-quantization scheme described here only recovers certain subtheories of the full quantum theory, which include only a subset of the preparations, transformations and measurements that the latter allows.  However, such derivations may provide clues for more ambitious axiomatization schemes.  
In particular, it provides further evidence for the usefulness of foundational principles asserting a fundamental restriction on what can be known~\cite{caves2002quantum,zeilinger1999foundational,paterek2010theories}. 
It also seems to highlight the importance of symplectic structure, which is not currently a feature of any reconstruction program.





Furthermore, epistricted theories constitute a foil to the full quantum theory in two senses.  When they are operationally equivalent to a subtheory of the full quantum theory, they are still a foil in the sense that 
the universe might have been governed by this subtheory rather than the full theory.  
Why did nature choose the full theory rather than the subtheory?
When the epistricted theory is not operationally equivalent to the corresponding quantum subtheory, as in the case of bits, it is a foil not only to the full quantum theory, but to the quantum subtheory as well.  In this case, the epistricted theory describes a set of operational predictions that are not instantiated in our universe and again the question is: why nature did not avail itself of this option?

Because epistricted theories share so much of the operational phenomenology of quantum theory, they constitute points in the landscape of possible operational theories that are particularly close to quantum theory.  They are therefore particularly helpful in the project of determining what is unique about quantum theory.  For any purported attempt to derive the formalism of quantum theory from axioms, it is useful to ask which 
of the axioms rule out epistricted theories.  Axiom sets that may seem promising at first can often be ruled out immediately if they fail to pass this simple test. 









Finally, epistricted theories have significance for the problem of developing mathematical frameworks for describing the landscape of operational theories.   In particular, they provide a test of the \emph{scope} of any given framework.  
Broadness of scope is the key virtue of any framework
because an axiomatic derivation of quantum theory is only as  impressive as the size of the landscape within which it is derived.  
For instance, the formalism of $C^{*}$-algebras essentially includes only operational theories that are fully quantum or fully classical, or that are quantum within each of a set of superselection sectors and classical between these. This is a rather limited scope, and consequently axiomatizations within this framework are less impressive than those formulated within broader frameworks.

On this front, epistricted theories serve to highlight two deficiencies in the prevailing framework of convex operational theories.  

First, quadrature epistemic theories are an example of a \emph{possibilistic} or \emph{modal} theory, wherein one does not specify the \emph{probabilities} of measurement outcomes, but only which outcomes are \emph{possible} and which are \emph{impossible}.  This perspective on the toy theory of Ref.~\cite{spekkens2007evidence}, for instance, is emphasized in Ref.~\cite{coecke2011phase}.
  Possibilistic theories have recently been receiving renewed attention~\cite{mansfield2012hardy, abramsky2012logical,schumacher2012modal} because in the context of discussions of Bell's theorem they highlight the fact that quantum theory cannot merely imply an innovation to probability theory but must also imply an innovation to logic. 

Second, these epistricted theories have operational state spaces that are not convex.  They only allow certain mixtures of operational states. 
As such, they  cannot be captured by the prevailing framework of \emph{convex} operational theories~\cite{barrett2007information,hardy2001quantum} because these assume from the outset a convex state space.  On the other hand, epistricted theories can be captured by the category-theoretic framework for process theories~\cite{Coecke2009}, as shown in \cite{coecke2011toy}, or by the framework of general probabilistic theories described in Ref.~\cite{chiribella2010probabilistic}.  Epistricted theories therefore provide a concrete example of how the category-theoretic framework necessarily describes some real estate in the landscape of foil theories that is not on the map of the convex operational framework.

\section{Quadrature epistricted theories}

\subsection{Classical complementarity as an epistemic restriction}\label{complementarity}

The criterion on the joint knowability of classical variables that is used here is
inspired by the criterion on the \emph{joint measurability} of quantum observables. 

\begin{quote}
{\bf Guiding analogy}: \\
A set of observables is \emph{jointly measurable} if and only if it is commuting relative to the matrix commutator. \\
A set of variables is \emph{jointly knowable} if and only if it is commuting relative to the Poisson bracket. \\
\end{quote}

The full epistemic restriction that we adopt is a combination of this notion of joint knowability together with the restriction that the only variables that can be known by the agent are \emph{linear} combinations of the position and momentum variables.    We refer to such variables as \emph{quadrature variables}.\footnote{This terminology comes from optics, where it was originally used to describe a pair of variables that are canonically conjugate to one another.  It was inherited from the use of the expression in astronomy, where it applies to a pair of celestial bodies and describes the configuration in which they have an angular separation of 90 degrees as seen from the earth.}
We term the full epistemic restriction \emph{classical complementarity}.
\begin{quote}
{\bf Classical complementarity}: 
The valid epistemic states are those wherein an agent knows the values of a set of quadrature variables that commute relative to the Poisson bracket and is maximally ignorant otherwise.
\end{quote}




It is presumed that maximal ignorance corresponds to a probability distribution that is uniform over the region of phase-space consistent with the known values of the quadrature variables.  In the case of a phase-space associated to a continuous field, uniformity is evaluated relative to the measure that is invariant under phase-space displacements.
Hence a valid epistemic state is a uniform distribution over the ontic states that is consistent with a given valuation of some Poisson-commuting set of quadrature variables. 
 It is because of the uniformity of these distributions that they can be understood as merely specifying, for the given constraints,  which ontic states are possible and which are
impossible.  Consequently, the epistemic state in this case is aptly described as a \emph{possibilistic} state. 

There is a subtlety here. \ The epistemic restriction is assumed to apply \emph{only} to what an agent
can know about a set of variables based on information acquired entirely to
the past or entirely to the future of those variables.  It is not assumed to apply to what an agent can know about a set of variables based on pre- and post-selection. 
The same caveats on applicability hold for the quantum uncertainty principle, so this constraint on applicability is not unexpected.

To describe the epistemic restriction in more detail, we introduce some formalism.  The continuous and discrete cases are considered in turn.

\textbf{Continuous degrees of freedom.}
 Assume $n$ classical continuous degrees of freedom. \ The configuration
space is $\mathbb{R}^{n}$ and a particular configuration is denoted
\begin{equation}
({\tt q}_{1},{\tt q}_{2},\dots,{\tt q}_{n})\in \mathbb{R}^{n}.
\end{equation}
These could describe the positions of $n$ particles in a 1-dimensional space, or the positions of $n/3$ particles in a 3-dimensional space, or the amplitudes of $n$ scalar fields,
etcetera.  The associated phase space is%
\begin{equation}
\Omega\equiv\mathbb{R}^{2n}
\end{equation}
and we denote a point in this space by 
\begin{equation}
\textbf{m} \equiv ({\tt q}_{1},{\tt p}_{1},{\tt q}_{2},{\tt p}_{2},\dots,{\tt q}_{n},{\tt p}_{n})\in \Omega.
\end{equation}

We consider real-valued functionals over this phase space
\begin{equation}
f:\Omega\rightarrow\mathbb{R}\text{.}%
\end{equation}
In particular, the functionals associated with the position and momentum of
the $i$th degree of freedom are defined respectively by%
\begin{equation}
q_{i}({\bf m})={\tt q}_{i},\text{ }p_{i}({\bf m})={\tt p}_{i}.
\end{equation}
The Poisson bracket is a binary operation on a pair of functionals, defined
by\
\begin{equation}
\left[  f,g\right]_{\rm PB}  \left(  {\bf m}\right)  \equiv\sum_{i=1}^{n}\left(
\frac{\partial f}{\partial q_{i}}\frac{\partial g}{\partial p_{i}}%
-\frac{\partial f}{\partial p_{i}}\frac{\partial g}{\partial q_{i}}\right)
\left(  {\bf m}\right)  .
\end{equation}
In particular, we have
\begin{equation}\label{eq:CCRcontinuous}
\left[ q_i,p_j\right]_{\rm PB}  \left(  {\bf m}\right) = \delta_{i,j}.
\end{equation}


The assumption of classical complementarity incorporates a restriction on the sorts of functionals that an agent can know.  Specifically, an agent can only know the value of a functional that is \emph{linear} in
the position and momentum functionals, that is, those of the form%
\begin{equation}
f={\tt a}_{1}q_{1}+{\tt b}_{1}p_{1}+\dots+{\tt a}_{n}q_{n}+{\tt b}_{n}p_{n} +{\tt c},
\label{eq:linearfunctionals}%
\end{equation}
where ${\tt a}_{1},{\tt b}_{1},\dots,{\tt a}_{n},{\tt b}_{n},{\tt c} \in\mathbb{R}$.  (Note that functionals that differ only by addition of a scalar or by a multiplicative factor ultimately describe the same property.)
We will call these \emph{quadrature functionals} or \emph{quadrature variables}.   The vector of coefficients of the position and momentum functionals for a given quadrature functional will be denoted by the boldface of the
notation used for the functional itself.  The vector ${\bf f}$ specifying the position and momentum dependence of the quadrature functional $f$ defined in Eq.~\eqref{eq:linearfunctionals} is
\begin{equation}
{\bf f} \equiv ({\tt a}_{1},{\tt b}_{1},\dots,{\tt a}_{n},{\tt b}_{n}),
\end{equation}
such that if we define the vector of position and momentum functionals 
\begin{equation}\label{vectorqpfunctionals}
{\bf z} \equiv (q_{1},p_{1},\dots,q_{n},p_{n}),
\end{equation}
we can express $f$ as
\begin{equation}\label{eq:f}
f= {\bf f}^T {\bf z} + {\tt c}.
\end{equation}
Similarly, the action of the functional $f$ on a phase space vector ${\bf m}$ is given by
\begin{equation}
f({\bf m}) = {\bf f}^T {\bf m}+ {\tt c}.
\end{equation}
In other words, the space of quadrature functionals is the dual of the phase space $\Omega$, but each functional $f$ is associated with a vector in the phase space, ${\bf f} \in \Omega$.
Note that the vectors associated with the position and momentum functionals $q_i$ and $p_i$ are ${\bf q}_i\equiv (0,0,\dots,1,0,\dots,0,0)$ where the only nonzero component is ${\tt a}_i$ and ${\bf p}_i \equiv (0,0,\dots,0,1,\dots,0,0)$, where the only nonzero component is ${\tt b}_i$.

It is not difficult to see that the Poisson bracket of two quadrature functionals always evaluates to a functional that
is uniform over the phase space.  Its value is equal to the symplectic inner product
of the associated vectors,
\begin{equation}
\left[ f,g\right]  _{PB}({\bf m})=\langle {\bf f}, {\bf g} \rangle, \label{eq:PoissonSymplecticIP}%
\end{equation}
where
\begin{equation}
\langle {\bf f}, {\bf g} \rangle \equiv {\bf f}^{T}J{\bf g}, 
\end{equation}
with $T$ denoting transpose and $J$ denoting the skew-symmetric $2n\times 2n$ matrix with components $J_{ij} \equiv\delta _{i,j+1}-\delta _{i+1,j}$, that is,
\begin{equation}
\label{eq:SymplecticForm}
J \equiv
\begin{pmatrix}
0 & 1 &  0 & 0 & \dots \\
-1 & 0 &  0 & 0 & \\
0 & 0 & 0  & 1 & \\
0 & 0 & -1 & 0 &  \\
 \vdots & &  & & \ddots
\end{pmatrix}.
\end{equation}
(Note that $J$ squares to the negative of the $2n \times 2n$ identity matrix, $J^2 = -I$, it is an orthogonal matrix, $J^T J =I$,  it has determinant +1, and it has an inverse given by $J^{-1} = J^T =-J$.)
For instance, for $\Omega =\mathbb{R}^2$, if
${\bf f}=\left(  {\tt a},{\tt b}\right)  $ and ${\bf g}=\left(  {\tt a}^{\prime},{\tt b}^{\prime}\right)  ,$ then
$\langle {\bf f}, {\bf g} \rangle = 
{\tt a}{\tt b}^{\prime}-{\tt b}{\tt a}^{\prime}.$

The symplectic inner product on a phase space $\Omega$
is a bilinear form $\langle \cdot ,\cdot \rangle :\Omega  \times \Omega \rightarrow \mathbb{R}$ that is skew-symmetric ($\langle {\bf f}, {\bf g} \rangle = -\langle {\bf g}, {\bf f} \rangle $ for all ${\bf f}, {\bf g} \in \Omega$) and non-degenerate (if $\langle {\bf f}, {\bf g} \rangle  =0$ for all ${\bf g} \in \Omega,$ then ${\bf f}=0$). \ By equipping the vector space $\Omega$ with the
symplectic inner product, it becomes a symplectic vector space.  This connection to symplectic geometry allows us to provide a simple geometric interpretation of the Poisson-commuting sets of quadrature functionals, which we will present in Sec.~\ref{validepistemicstates}.



\textbf{Discrete degrees of freedom.} For discrete degrees of freedom, the formalism is precisely the same, except that
variables are no longer valued in the real field $\mathbb{R}$, but a finite field instead.  Recall that all finite fields have order equal to the power of a prime.  We shall consider here only the case where the order is itself a prime, denoted $d$, in which case the field is isomorphic to the integers modulo $d$, which we will denote by $\mathbb{Z}_{d}$. Therefore, the configuration space is $\left(  \mathbb{Z}_{d}\right)^n$, 
the associated phase space is 
\[
\Omega\equiv\left(  \mathbb{Z}_{d}\right)  ^{2n},
\]
the functionals have the form 
\[
f:\Omega\rightarrow  \mathbb{Z}_{d}  \text{,}
\]
and the linear functionals are of the form of Eq. (\ref{eq:linearfunctionals}),
\begin{equation}
f={\tt a}_{1}q_{1}+{\tt b}_{1}p_{1}+\dots+{\tt a}_{n}q_{n}+{\tt b}_{n}p_{n} + {\tt c},
\end{equation}
but where ${\tt a}_{1},{\tt b}_{1},\dots,{\tt a}_{n},{\tt b}_{n},{\tt c} \in\mathbb{Z}_d$ and
the sum denotes addition modulo $d$.
 It follows that the vector ${\bf f} \equiv ({\tt a}_1, {\tt b}_1, \dots, {\tt a}_n, {\tt b}_n)$ associated with the functional $f$ lives in the phase space $\Omega \equiv\left(  \mathbb{Z}_{d}\right)^n$ as well.

 The Poisson bracket, however, cannot be defined
in the conventional way, because without continuous variables we do not have a
notion of derivative. \ \ Nonetheless, one can define a discrete version of
the Poisson bracket in terms of finite differences.  For any functionals $f: \Omega \to \mathbb{Z}_d$ and $g: \Omega \to \mathbb{Z}_d$, their Poisson bracket, denoted $[f,g]_{PB}$, is also such a functional, the one defined by
\begin{equation}
\left[  f,g\right]  _{PB}({\bf m})\equiv\sum_{i=1}^{n}\left[
\begin{array}[c]{c}
\left(  f\left(  {\bf m}+{\bf q}_{i}\right)  -f\left(  {\bf m}\right)  \right)  \left(
g\left(  {\bf m} +{\bf p}_{i} \right)  -g\left(  {\bf m}\right)  \right) \\
-\left(  f\left(  {\bf m} +{\bf p}_{i}\right)  -f\left(  {\bf m}\right)  \right)  \left(
g\left( {\bf m}+{\bf q}_{i}\right)  -g\left(  {\bf m}\right)  \right)
\end{array}
\right],
\end{equation}
where the differences in this expression are evaluated with modular arithmetic.
The requirement that
\begin{equation}\label{eq:CCRdiscrete}
\left[ q_i,p_j\right]_{\rm PB}  \left(  {\bf m}\right) = \delta_{i,j},
\end{equation}
is clearly satisfied.  Furthermore, it is straightforward to verify that Eq.~(\ref{eq:PoissonSymplecticIP}) also
holds under this definition, so that one can relate the Poisson bracket in the discrete setting to the symplectic inner
product on the discrete phase space, $\langle \cdot ,\cdot \rangle :\Omega  \times \Omega \rightarrow \mathbb{Z}_d$.



\textbf{Simple examples.}  
It is useful to consider some simple examples of commuting pairs of quadrature variables, that is, some examples of 2-element sets $\left\{
f,g\right\}  $ such that $\left[  f,g\right]  _{PB}=0.$ \ Any quadrature variable defined
for system 1 commutes with any quadrature variable for system 2, e.g., the pair%
\begin{equation}
{\tt a}_{1}q_{1}+{\tt b}_{1}p_{1},\;\;\;{\tt a}_{2}q_{2}+{\tt b}_{2}p_{2}
\end{equation}
is a commuting pair for any values of ${\tt a}_{1},{\tt b}_{1},{\tt a}_{2},{\tt b}_{2}\in \mathbb{R}$ (or ${\tt a}_{1},{\tt b}_{1},{\tt a}_{2},{\tt b}_{2}\in\mathbb{Z}_d$).
 Additionally, there are commuting pairs of quadrature variables describing joint
properties of the two systems, for instance%
\begin{equation}
q_{1}-q_{2},\;\;\;p_{1}+p_{2}
\end{equation}
(when the field is $\mathbb{Z}_d$, the coefficient $-1$ is equivalent to $d-1$, so that $q_{1}-q_{2}=q_{1}+(d-1)q_{2}$). 

Another useful concept in the following will be that of \emph{canonically conjugate}
variables. A pair of variables are said to be canonically conjugate if
$\left[  f,g\right]  _{PB} = 1.$\ On a
single system, the pair of quadrature variables%
\begin{equation}
  {\tt a}q+{\tt b}p,\;\;\;-{\tt b}q+{\tt a}p
\end{equation}
are canonically conjugate for any values ${\tt a},{\tt b} \in \mathbb{R}$ (or ${\tt a},{\tt b} \in \mathbb{Z}_d$) such that ${\tt a}^2 + {\tt b}^2 =1$; in particular $\{q,p\}$ is such a pair.

Note that we were able to present these examples without specifying the nature of the field.  We will follow this convention of presenting results in a unified field-independent manner for the next few sections.



\subsection{Characterization of quadrature epistricted theories}

\subsubsection{The set of valid epistemic states}\label{validepistemicstates}

Using the connection between the Poisson bracket for quadrature functionals and the
symplectic inner product, one obtains a geometric interpretation of the
epistemic restriction and the valid epistemic states.

To specify an epistemic state one must specify: (i) the set of quadrature variables that are known
to that agent and (ii) the values of these variables.   We will consider each aspect in turn.

The epistemic restriction asserts that the only sets of variables that are jointly knowable are those that are Poisson-commuting (which is to say that every pair of elements in the set Poisson-commutes).
Note, however, that if every variable in a set has a known value, then any function of those variables also has a known value, in particular any linear combinations of those variables has a known value.  It follows that for any Poisson-commuting set of variables, we can close the set under linear combination and preserve the property of being Poisson-commuting.  In terms of the vectors representing these variables, this implies that we can take their \emph{linear span} while preserving the property of having vanishing symplectic inner product for every pair of vectors.   In terms of the symplectic geometry, a subspace all of whose vectors have vanishing symplectic inner product with one another is called an \emph{isotropic} subspace of the phase space.  Formally, a subspace $V\subseteq \Omega$ is isotropic if 
\begin{equation}
\forall {\bf f},{\bf g}\in V: \langle {\bf f}, {\bf g} \rangle =0.
\end{equation}
It follows that we can parametrize the different possible sets of known variables in terms of the isotropic subspaces of the phase space $\Omega$. 

For a $2n$-dimensional phase space, the maximum possible dimension of an isotropic subspace is $n$.  These are called \emph{maximally isotropic} or \emph{Lagrangian} subspaces.  This case corresponds to the maximal possible knowledge an agent can have according to the epistemic restiction.  The agent then knows a \emph{complete} set of Poisson-commuting variables, which is the analogue of measuring a complete set of commuting observables in quantum theory.  

For a given Poisson-commuting set of variables, define a basis of that set to be any subset containing linearly independent elements and from which the entire set can be obtained by linear combinations.  In the symplectic geometry, this corresponds to a vector basis for the associated isotropic subspace.   There are, of course, many choices of bases for a given isotropic subspace or Poisson-commuting set.

Next, we must characterize the possible value assignments to a Poisson-commuting set of quadrature variables.  That is, we must specify a linear functional $v$ acting on a quadrature functional $f$ 
and taking values in the appropriate field (continuous or discrete) such that $v(f)$ is the value assigned to $f$.  Denote the isotropic subspace of $\Omega$ that is associated to this Poisson-commuting set by $V$, such that ${\bf f} \in V$ is the vector associated with the quadrature function $f$. 
The set of value assignments corresponds precisely to the set of vectors in $V$.
  In other words,  for every vector ${\bf v} \in V$, which we call a  \emph{valuation vector}, we obtain a distinct value assignment $v$, via
\[
v(f) =  {\bf f}^T {\bf v}.
\]
To see this, it suffices to note that the ontic state of the system determines the values of all functionals and therefore the set of possible value assignments is given by the set of possible ontic states.  Specifically, each ontic state ${\bf m} \in \Omega$ defines the value assignment
\[
v_{\bf m}( f) =  {\bf f}^T {\bf m}.
\]
However, many different ontic states yield the same value assignment.  
Denoting the projector onto $V$ by $P_{V}$, we can express the relevant equivalence relation thus: the ontic states ${\bf m}$ and ${\bf m'}$ yield the same value assignment to the quadrature functionals associated to $V$ if and only if $P_{V} {\bf m} = P_{V} {\bf m'}$.  It follows that the set of possible value assignments to the associated to $V$ can be parametrized by the set of projections of all ontic states ${\bf m}\in \Omega$ into $V$, which is simply the set of ontic states in $V$.  This establishes what we set out to prove.

As an example, consider the case where we have two degrees of freedom, so that $\Omega$ is 4-dimensional, and suppose that the set of quadrature variables that are jointly known are the position variables, $\{ q_1, q_2\}$, and that these are known to each take the value $1$.  In this case, the associated isotropic subspace $V \subseteq \Omega$,
and the valuation vector ${\bf v}\in V$ are,  respectively,
\begin{align}
V  &  =\mathrm{span}\{{\bf q}_{1},{\bf q}_{2}\}\\
&= \mathrm{span}\{ (1,0,0,0), (0,0,1,0)\}\\
&=\{ ({\tt s},0,{\tt t},0): {\tt s},{\tt t} \in \mathbb{R}/\mathbb{Z}_d \}\\
{\bf v}  &  =(1,0,1,0).
\end{align}
These are depicted in green in Fig.~\ref{fig:Green1}.

\begin{figure}[h!]
 \center{   \includegraphics[width=3cm]{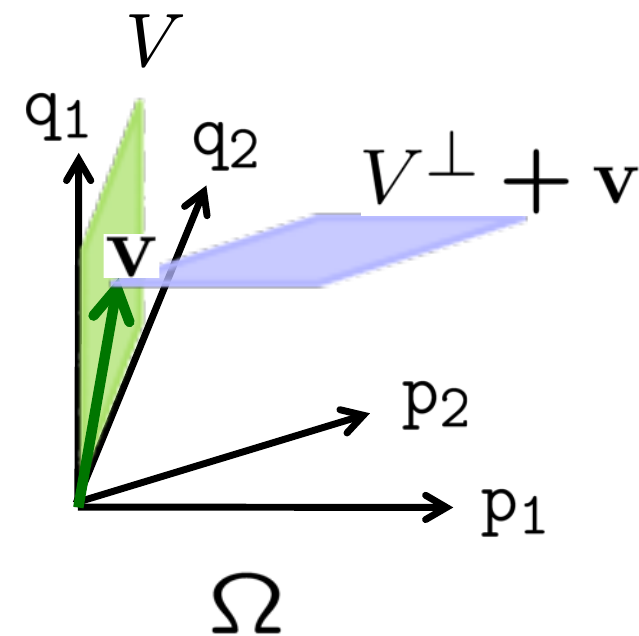}}
\caption{
A schematic of the 4-dimensional phase space of ontic states, $\Omega$, the isotropic subspace $V\subseteq \Omega$ associated with the known quadrature variables, the valuation vector ${\bf v}\in V$ specifying the values of the known variables and the manifold $V^{\perp} + {\bf v}$ corresponding  to the ontic support of the associated epistemic state.
}
 \label{fig:Green1}
\end{figure}



Next, we consider a given epistemic state, where the known quadrature variables are specified by an isotropic subspace $V \subseteq \Omega$, 
and their values are specified by ${\bf v} \in V$, 
and ask: what probability distribution over the phase space $\Omega$ does it correspond to?  
Recalling that this probability distribution should be maximally uninformative relative to the given constraint, 
 the answer is simply a uniform distribution on the set of all ontic states that yield this value assignment, that is, on the set
\begin{align}\label{setofonticstates}
&\{ {\bf m}\in \Omega :  {\bf f}^T {\bf m} = {\bf f}^T {\bf v} \;\forall {\bf f} \in V \} \nonumber \\
&  =\left\{  {\bf m}\in\Omega:P_{V}{\bf m}={\bf v}\right\}.
\end{align}
If we denote the subspace of $\Omega$ that is orthogonal to $V$ (relative to the Euclidean
inner product) by $V^{\perp},$ and we denote the translation of a subspace $W$
by a vector ${\bf v}$ as $W+{\bf v}\equiv\left\{  {\bf m}:{\bf m}={\bf w}+{\bf v},\text{ } {\bf w}\in W\right\}  ,$ then
it is clear that the set of ontic states of \eqref{setofonticstates} is simply
\[
V^{\perp}+{\bf v}.
\]

For instance, in the example described above, 
\begin{align}
V^{\perp}  
&= \{ (0,{\tt s},0,{\tt t}): {\tt s},{\tt t} \in \mathbb{R}/\mathbb{Z}_d \},
\end{align}
and consequently the set of ontic states consistent with the agent's knowledge is
\begin{equation}
V^{\perp}+{\bf v} \equiv \left\{  \left(  1,{\tt s},1,{\tt t}\right)  :{\tt s},{\tt t} \in\mathbb{R}/\mathbb{Z}
_{d}\right\},
\end{equation}
which is depicted in blue in Fig.~\ref{fig:Green1}.  

As a probability distribution over $\Omega$, the epistemic state associated to $\left(  V,{\bf v}\right)$ has the following form
\begin{equation}\label{epiprod}
\mu_{V,{\bf v}}\left(  {\bf m}\right)  = \frac{1}{\mathcal{N}_V } \delta_{V^{\perp}+{\bf v}}({\bf m})
\end{equation}
where we have introduced the notation
\begin{equation}\label{defndelta}
\delta_{V^{\perp}+{\bf v}}({\bf m}) \equiv \prod_{{\bf f}^{(i)}: {\rm span}\{ {\bf f}^{(i)} \} =V} \delta ( {\bf f}^{(i)T} {\bf m} - {\bf f}^{(i)T}{\bf v}), 
\end{equation}
where in the discrete case $\delta ({\tt c})=1$ if ${\tt c}=0$ and $\delta ({\tt c})=0$ otherwise, while  in the continuous case $\delta$ denotes a Dirac delta function.   In this expression, $\{ {\bf f}^{(i)}  \}$ can be any basis of $V$.
Geometrically, $\mu_{V,{\bf v}}$ is simply the uniform distribution over the ontic states  in $V^{\perp}+{\bf v}$.


Some epistemic states are seen to be mixtures of others in this theory.  A valid epistemic state is termed \emph{pure} if it is convexly
extremal among valid epistemic states, that is, if it cannot be formed as a
convex combination of other valid opistemic states. \ Non-extremal epistemic
states are termed \emph{mixed}. \ Note that we are judging extremality
relative to the set of valid epistemic states, not relative to the set of all epistemic states.  
 In our approach, the pure epistemic states are those corresponding to maximal knowledge, that is, knowledge associated to a complete set of Poisson-commuting quadrature variables.  Note, however, that because of the epistemic restriction, maximal knowledge is always incomplete knowledge.  

\subsubsection{The set of valid transformations}

In addition to specifying the valid epistemic states, we
must also specify what transformations of the epistemic states are allowed in our theory.  
To begin with, we consider the reversible transformations on an isolated system. 

Suppose an agent knows the precise ontological dynamics of a system over some period of time.  This transformation is represented by a bijective map on the ontic state space, and this induces a bijective map on the space of epistemic states.  

Because we assume that the underlying ontological theory has symplectic structure, it follows that the allowed transformations must be within the set of \emph{symplectic transformations} (sometimes called \emph{symplectomorphisms}).  
The requirement that the epistemic restriction must be preserved under the transformation implies that the valid transformations are a subset of the symplectic transformations, namely, those that map the set of quadrature variables to itself.  
Each such transformation can be represented in terms of its action on the
phase space vector ${\bf m}\in\Omega$\ as
\begin{equation}\label{symplecticaffine}
{\bf m}\mapsto S{\bf m}+{\bf a}
\end{equation}
where ${\bf a} \in\Omega$ is a phase-space displacement vector and where $S$ is a $2n \times 2n$ \emph{symplectic
matrix}, that is, one which preserves the symplectic form $J$ defined in Eq.~\eqref{eq:SymplecticForm},
\begin{equation}
S^{T}JS=J,
\end{equation}
or equivalently, one which preserves symplectic inner products, i.e., $\left(
S{\bf m}\right)^{T}J\left(  S{\bf m}^{\prime}\right)  ={\bf m}^{T}J{\bf m}^{\prime}\;\forall {\bf m},{\bf m}' \in \Omega$.
These are combinations of phase-space rotations and phase-space displacements.   



Equation~\eqref{symplecticaffine} describes an affine transformation, but it does not include all such transformations because $S$ is not a general linear matrix.  Following \cite{gross2006hudson}, we call transformations of the form of Eq.~\eqref{symplecticaffine} \emph{symplectic affine transformations}.  
Two such transformations, $S \cdot +{\bf a}$ and $S' \cdot +{\bf a}'$ compose as
\begin{equation}\label{sympaffine}
S \left(  S' \cdot +{\bf b'} \right) +{\bf b}  = SS' \cdot + (S{\bf b}' + {\bf b}).
\end{equation}
The inverse of a symplectic matrix $S$ is $S^{-1} = J^T S^T J$, and the inverse of the phase-space displacement ${\bf a}$ is of course $-{\bf a}$.
We call the resulting group of transformations the \emph{symplectic affine group}.



If the epistemic state is described by a probability distribution/density over ontic states, $\mu : \Omega \to \mathbb{R}_+$, then under the ontological transformation
${\bf m} \mapsto S{\bf m} + {\bf a}$, the transformation induced on the epistemic state is 
\begin{equation}
\mu ({\bf m}) \mapsto \mu'({\bf m}) = \mu (S^{-1}{\bf m} -{\bf a}),
\end{equation}
\color{black}
We can equivalently represent this transformation by a conditional probability distribution $\Gamma_{S,{\bf a}}: \Omega \times \Omega \to  \mathbb{R}_+$, that is,
\begin{equation}
\mu'({\bf m}) = \int {\rm d}{\bf m}' \Gamma_{S,{\bf a}}({\bf  m}|{\bf m}') \mu({\bf m}'),
\end{equation}
where
\begin{equation}\label{Cliffordcond}
\Gamma_{S,{\bf a}}({\bf  m}|{\bf m}') = \delta({\bf m} - (S{\bf m}' + {\bf a})).
\end{equation}

There is a subtlety worth noting at this point. The map $\mu \mapsto \mu'$ on the space of probability distributions, which is induced by the map ${\bf m} \mapsto S{\bf m} + {\bf a}$ on the space of ontic states, has the following property: it maps the set of valid epistemic states (those satisfying the classical complementarity principle) to itself.  However, not every map from the set of valid epistemic states to itself can be induced by some map on the space of ontic states.  A simple counterexample is provided by the map corresponding to time reversal. \ For a single degree of freedom, time reversal is represented by the map ${\bf m} = ({\tt q},{\tt p}) \mapsto {\bf m}' = ({\tt q},-{\tt p})$, 
which obviously fails to preserve the symplectic form.  In terms of symplectic geometry, it is a reflection rather than a rotation in the phase space.  Nonetheless, it maps isotropic subspaces to isotropic subspaces and therefore it also maps valid epistemic states to valid epistemic states.  Therefore, in considering a given map on the space of distributions over phase space, it is not sufficient to ensure that it takes valid epistemic states to valid epistemic states, one must also ensure that it arises from a possible ontological dynamics.  We say that the map must \emph{supervene} upon a valid ontological transformation~\cite{bartlett2012reconstruction}.

Note that if the phase space is over a {\em discrete} field, then the
transformations must be discrete in time. \ Only in the case of continuous
variables can the transformations be continuous in time and only in this case can
they be generated by a Hamiltonian. 

In addition to transformations corresponding to reversible maps over the epistemic states, there are also transformations corresponding to irreversible maps. 
These correspond to the case where information about the system is lost. 
The most general such transformation corresponds to adjoining the system to an ancilla that is prepared in a quadrature state, evolving the pair by some symplectic affine transformation that involves a nontrivial coupling of the two, and finally marginalizing over the ancilla.
The reason this leads to a loss of information about the ontic state of the system is that the transformation of the system depends on the initial ontic state of the ancilla, and the latter is never completely known, by virtue of the epistemic restriction.  



\subsubsection{The set of valid measurements}

We must finally address the question of which \emph{measurements} are consistent with our epistemic restriction.  We will distinguish sharp and unsharp measurements.  The sharp measurements are the analogues of those associated with projector-valued measures in quantum theory and can be defined as those for which the outcome is deterministic given the ontic state. The unsharp measurements are the analogues of those in quantum theory that cannot be represented by a projector-valued measure but instead require a positive operator-valued measure; they can be defined as those for which the outcome is not deterministic given the ontic state.  

We begin by considering the valid \emph{sharp} measurements.  
Without the epistemic restriction, one could imagine the possibility
of a sharp measurement that would determine the values of
\emph{all} quadrature variables, and hence also determine what the ontic state of the system was prior to the measurement. \ Given classical complementarity, however, one can only jointly retrodict the values of a set of quadrature variables  if these are a Poisson-commuting set, and therefore the only sets of  quadrature variables that can be jointly measured are the Poisson-commuting sets. 

Given that every Poisson-commuting set of quadrature variables defines an isotropic subspace, the valid sharp measurements are parametrized by the isotropic subspaces.  Furthermore, the possible joint value-assignments to a Poisson-commuting set of variables associated with isotropic subspace $V$ are parametrized by the vectors in $V$, so that the outcomes of the measurement associated with $V$ are indexed by ${\bf v} \in V$. 

Such measurements can be represented as a conditional probability, specifying the probability of each outcome ${\bf v}$ given the ontic state ${\bf m}$, namely, 
\begin{equation}\label{responsefns}
\xi_{V}\left({\bf v}|  {\bf m}\right)  =
\delta_{V^{\perp} + {\bf v}}({\bf m}),
\end{equation}
where $\delta_{V^{\perp}+{\bf v}}({\bf m})$ is defined in Eq.~\eqref{defndelta}.
We refer to 
 the set $\{ \xi_V ({\bf v} | {\bf m}): {\bf v}\in V\}$, considered as functions over $\Omega$, as the \emph{response functions} associated with the measurement.


The set of all valid {\em unsharp} measurements can then be defined in terms of the valid sharp measurements as follows.  An unsharp measurement on a system is valid if it can be implemented by adjoining to the system an ancilla that is described by a valid epistemic state, coupling the two by a symplectic affine transformation, and finally implementing a valid sharp measurement on the system+ancilla.  
Note that this construction of unsharp measurements from sharp measurements on a larger system is the analogue of the Naimark dilation in quantum theory.  

A full treatment of measurements would include a discussion of how the epistemic state is updated when the system survives the measurement procedure, but we will not discuss the transformative aspect of measurements in this article.

\subsubsection{Operational statistics}

Suppose that one prepares a system with phase space $\Omega$ in the epistemic state $\mu_{V,{\bf v}}({\bf m})$ associated with isotropic subspace $V$ and valuation vector ${\bf v}$, and one subsequently implements the sharp measurement associated with the isotropic subspace $V'$.  What is the probability of obtaining a given outcome ${\bf v}' \in V'$?  The answer follows from an application of the law of total probability.  The probability is simply
\begin{eqnarray}\label{opstat1}
&&\sum_{{\bf m} \in \Omega} \xi_{V'}({\bf v}' | {\bf m}) \mu_{V,{\bf v}}({\bf m}).\nonumber\\
\end{eqnarray} 
If a symplectic affine transformation ${\bf m} \mapsto S{\bf m} +{\bf a}$ is applied between the preparation and the measurement, the probability of outcome ${\bf v}'$ becomes
\begin{eqnarray}\label{opstat2}
&&\sum_{{\bf m} \in \Omega} \xi_{V'}({\bf v}' | {\bf m}) \sum_{{\bf m}' \in \Omega} \Gamma_{S,{\bf a}}({\bf m}|{\bf m}')  \mu_{V,{\bf v}}({\bf m}').
 \nonumber\\
\end{eqnarray} 

These statistics constitute the operational content of the quadrature epistricted theory.


\subsection{Quadrature epistricted theory of continuous variables}

We now turn to concrete examples of quadrature epistricted theories for particular choices of the field.
In this section, we consider the 
case of a phase space of $n$ {\em real} degrees of freedom,
$\Omega=\mathbb{R}^{2n}.$ 
We begin by discussing the valid epistemic states for a single degree of
freedom, $n=1$. \ In this case, the phase space is 2-dimensional and the isotropic subspaces are the set of 1-dimensional
subspaces. \ We have depicted a few examples in Fig.~\ref{fig:CVsingledofs}. \ The isotropic
subspace $V$ is depicted in light green, the valuation vector ${\bf v}$ is depicted as a
dark green arrow, and the set $V^{\perp}+{\bf v}$ of ontic states in the support of the associated
epistemic state is depicted in blue. \ Fig.~\ref{fig:CVsingledofs}(a) depicts a
state of knowledge wherein position is known (and hence momentum is unknown). Fig.~\ref{fig:CVsingledofs}(b) depicts the vice-versa. \ Fig.~\ref{fig:CVsingledofs}(c) corresponds to knowing the value of a quadrature $\left(  \cos\theta\right) q+\left(  \sin\theta\right) p$ (and hence
having no knowledge of the canonically conjugate quadrature $-\left(
\sin\theta\right)  q+\left(  \cos\theta\right)  p$).  Finally, an agent could know nothing at all, in which case the epistemic state
is just the uniform distribution over the whole phase space, as depicted in Fig.~\ref{fig:CVsingledofs}(d).

\begin{figure}[h!]
 \center{   \includegraphics[width=12cm]{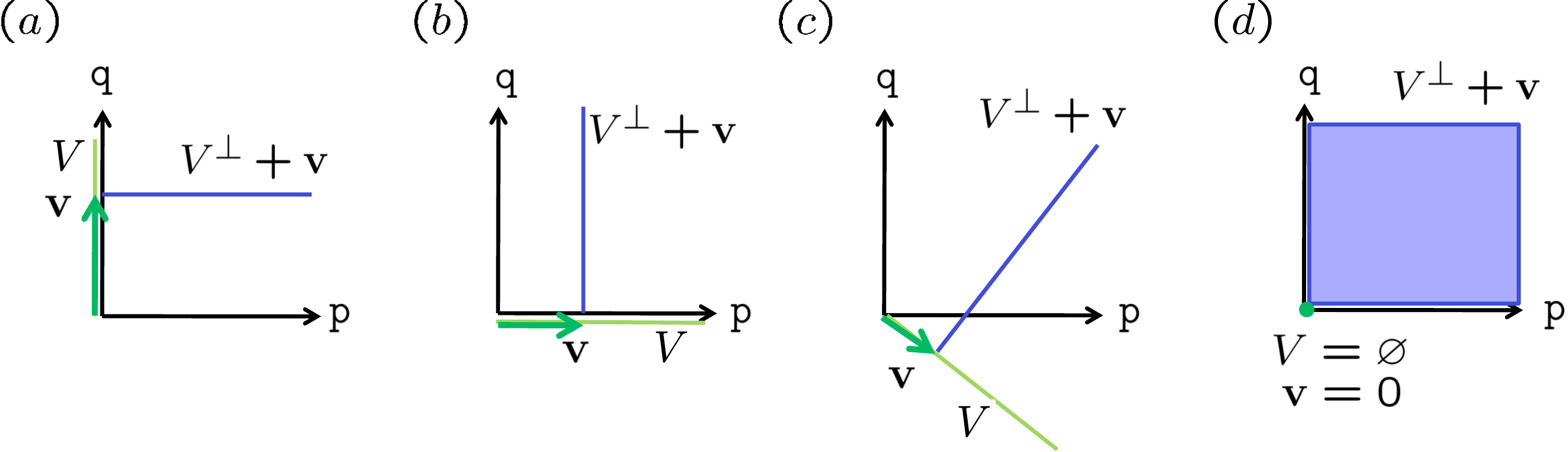}}
\caption{Examples of valid epistemic states for a single continuous variable system.}
\label{fig:CVsingledofs}
\end{figure}

If one considers a pair of continuous degrees of freedom, then it becomes
harder to visualize the epistemic states because the phase space is
4-dimensional. Nonetheless, we present 3-dimensional
projections as a visualization tool. \ We know that for every pair of isotropic subspace and valuation vector, $(V,{\bf v})$,
there is a distinct epistemic state. \ In Fig.~\ref{fig:CVtwodofs}(a), 
we depict the example where $q_1$ and $q_2$ are the known variables and both take the value 1, so that 
$V=\mathrm{span}\left\{  {\bf q}_{1},{\bf q}_{2}\right\}=\mathrm{span}\left\{  (1,0,0,0), (0,0,1,0)\right\}  $
 and ${\bf v}=(1,0,1,0)$,
while in Fig.~\ref{fig:CVtwodofs}(b), it is $q_{1}-q_{2}$ and $p_{1}
+p_2$ that are the known variables and both take the value 1, so that  
$V=\mathrm{span}\left\{  {\bf q}_{1}-{\bf q}_{2},{\bf p}_{1}+{\bf p}_{2}\right\}=\mathrm{span}\left\{  (1,0,-1,0), (0,1,0,1)\right\} $
 and ${\bf v}=\left( \tfrac{1}{2},\tfrac{1}{2},-\tfrac{1}{2},\tfrac{1}{2}\right)  .$\ \ In the example of Fig.~\ref{fig:CVtwodofs}(c),
only a single variable, $q_1$, is known and takes the value 1, so that  $V=\mathrm{span}\{{\bf q}_1 \}=\mathrm{span}\left\{ (1,0,0,0) \right\}  $ and
${\bf v}=(1,0,0,0)$.


\begin{figure}[h!]
 \center{   \includegraphics[width=12cm]{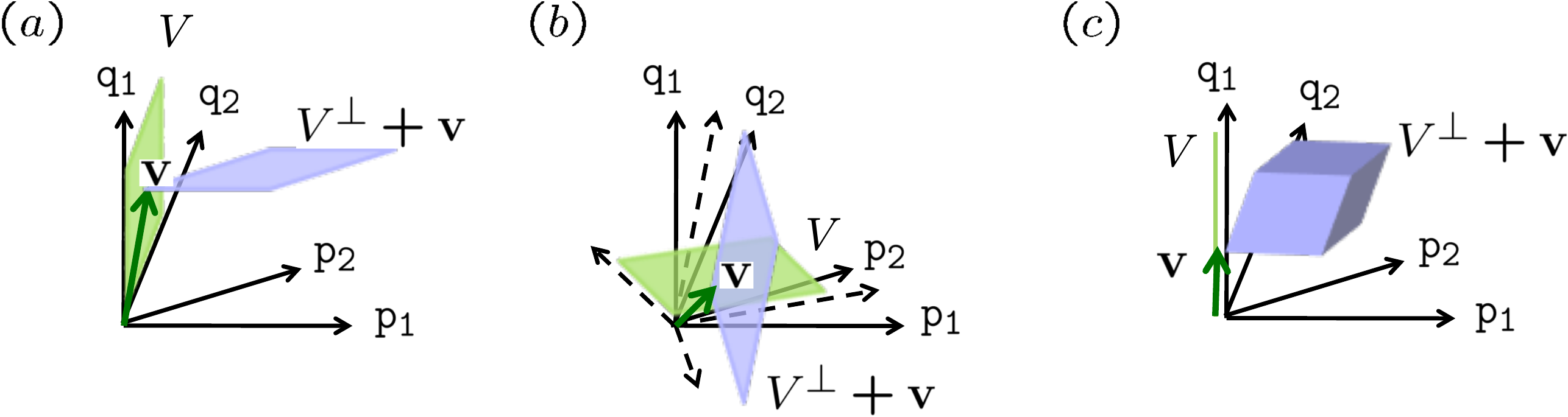}}
\caption{Examples of valid epistemic states for a pair of continuous variable systems.}
\label{fig:CVtwodofs}
\end{figure}


\subsection{Quadrature epistricted theory of trits}

We turn now to discrete systems. 
We begin with the case where the configuration space of
every degree of freedom is three-valued, i.e., a \emph{trit}, and represented therefore by $\mathbb{Z}_{3}$, the integers modulo 3.  The configuration space of $n$ degrees of freedom is $\mathbb{Z}_{3}^n$ and the phase space is $\Omega=\left(  \mathbb{Z}_{3}\right)  ^{2n}.$

For a single system ($n=1)$, we can depict $\Omega$ as a 3$\times3$ grid.  Consider
all of the quadrature functionals that can be defined on such a system.  They are of the form $f={\tt a}q+{\tt b}p+{\tt c}$ where ${\tt a},{\tt b},{\tt c} \in \mathbb{Z}_3$.  Some of
these functionals partition the phase-space in equivalent ways.
It suffices to look at the inequivalent quadrature functionals.  There are four of these:
\begin{equation}
q,\;p,\;q+p,\;q+2p.
\end{equation}
Note that because addition is modulo 3, $q+2p$ could
equally well be written $q-p.$

\begin{figure}[h!]
 \center{   \includegraphics[width=8cm]{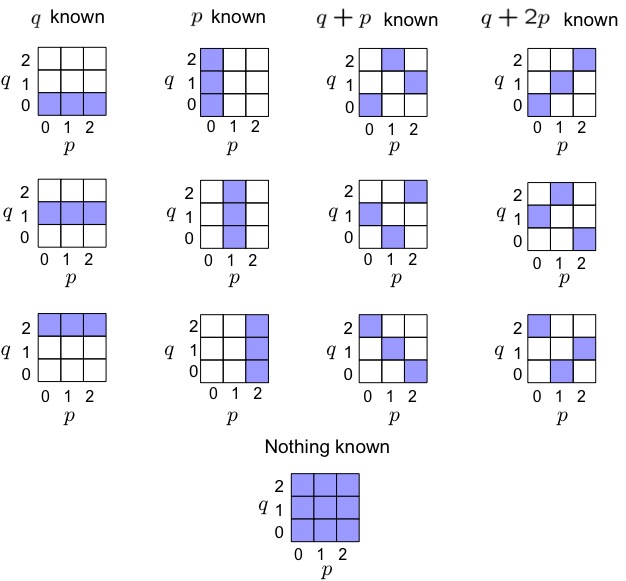}}
\caption{All the valid epistemic states for a single trit. There are twelve states of maximal knowledge (one variable known) and a single state of nonmaximal knowledge (no variable known).}
\label{fig:epistemicstates1toytrit}
\end{figure}

Because no two of these functionals Poisson-commute, the principle of classical complementarity implies that an agent can know the value of at most one of these variables. \ It follows that there are twelve pure epistemic states,
depicted in Fig.~\ref{fig:epistemicstates1toytrit}.  
The only mixed state is the state of complete ignorance.
Here we depict in blue the ontic
states in the support of the epistemic state. \ We have not explicitly
depicted the isotropic subspace and valuation vector, but these are analogous to what we had in the continuous variable case.

Next, we can consider \emph{pairs} of trits $(n=2)$. The quadrature variables are linear
combinations of the positions and momentum of each, with coefficients drawn
from $\mathbb{Z}_3$. Just as in the continuous case, one now has quadrature
variables that describe joint properties of the pair of systems.
The complete sets of Poisson-commuting variables now contain a pair of variables. 
Rather than attempting to portray the 4-dimensional phase space, as we did
in the continous case, we can depict each 2-dimensional symplectic subspace
along a line, as in Fig.~\ref{fig:Sudokupuzzle}. This is the \textquotedblleft Sudoku
puzzle\textquotedblright\ depiction of the two-trit phase space.

\begin{figure}[h!]
 \center{   \includegraphics[width=8cm]{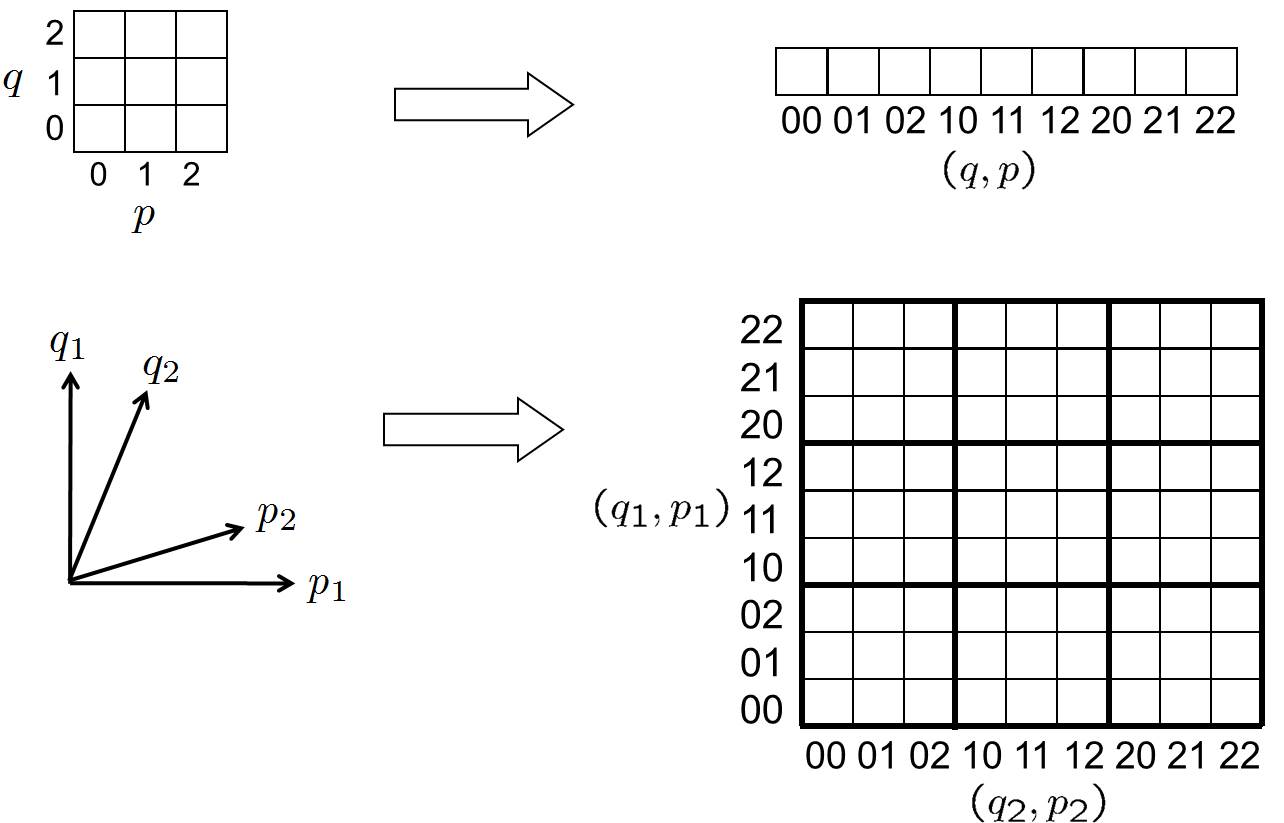}}
\caption{Embedding a 4-dimensional discrete phase space in a 2-dimensional grid resembling a Sudoku puzzle.}
\label{fig:Sudokupuzzle}
\end{figure}

Figs.~\ref{fig:twotoytrits}(a), \ref{fig:twotoytrits}(b), \ref{fig:twotoytritsentangled}(a), and \ref{fig:twotoytritsentangled}(b) each depict a mixed epistemic state, wherein the value of a single
quadrature variable is known. \ Figs.~\ref{fig:twotoytrits}(c) and  ~\ref{fig:twotoytritsentangled}(c) depict pure epistemic states, wherein the
values of a pair of Poisson-commuting variables are known.  If one of the pair of known variables refers to the first subsystem and the other refers to the second subsystem, as in Fig.~\ref{fig:twotoytrits}(c), the epistemic state corresponds to a product state in quantum theory.  If both of the known variables describe joint properties of the pair of trits, as in Fig.~\ref{fig:twotoytritsentangled}(c), the epistemic state corresponds to an entangled state.

\begin{figure}[h!]
 \center{   \includegraphics[width=8cm]{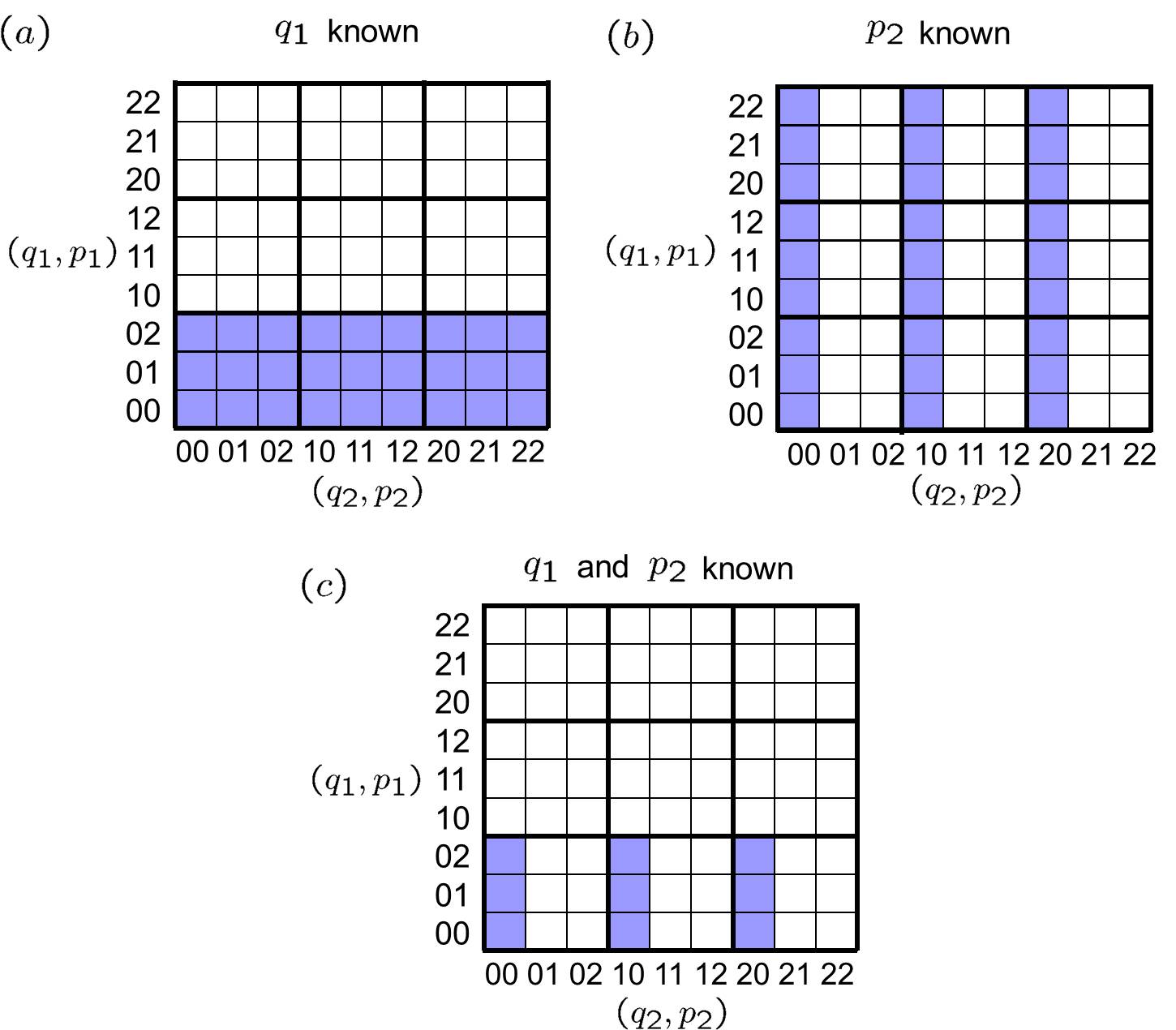}}
\caption{Examples of epistemic states for two trits.  (a) and (b) are examples where one variable is known. (c) is an example where two variables are known, corresponding to a product quantum state.}
\label{fig:twotoytrits}
\end{figure}

\begin{figure}[h!]
 \center{   \includegraphics[width=8cm]{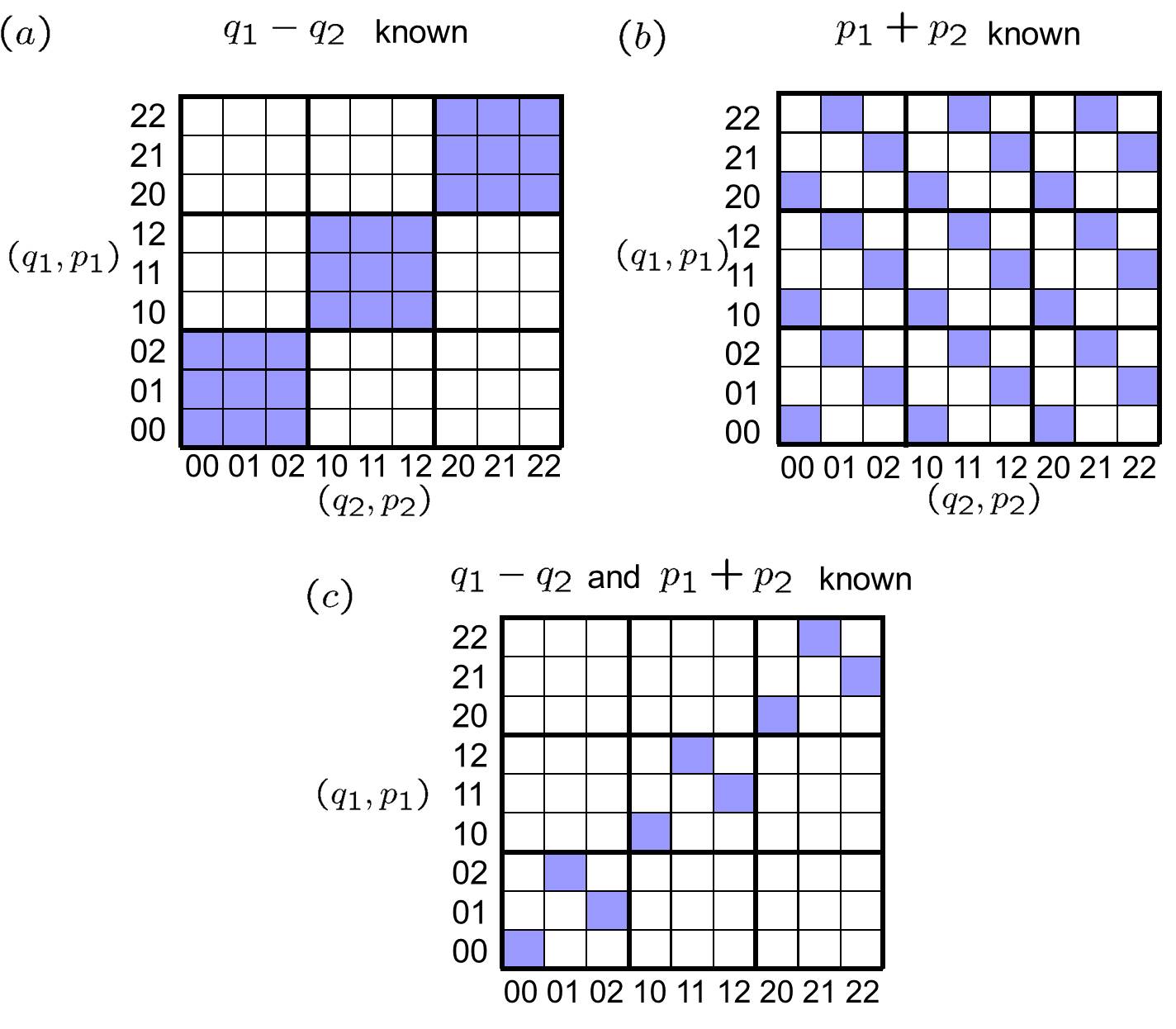}}
\caption{Examples of epistemic states for two trits.  (a) and (b) are examples where one variable is known. (c) is an example where two variables are known, corresponding to an entangled quantum state}
\label{fig:twotoytritsentangled}
\end{figure}

\begin{figure}[h!]
 \center{   \includegraphics[width=8cm]{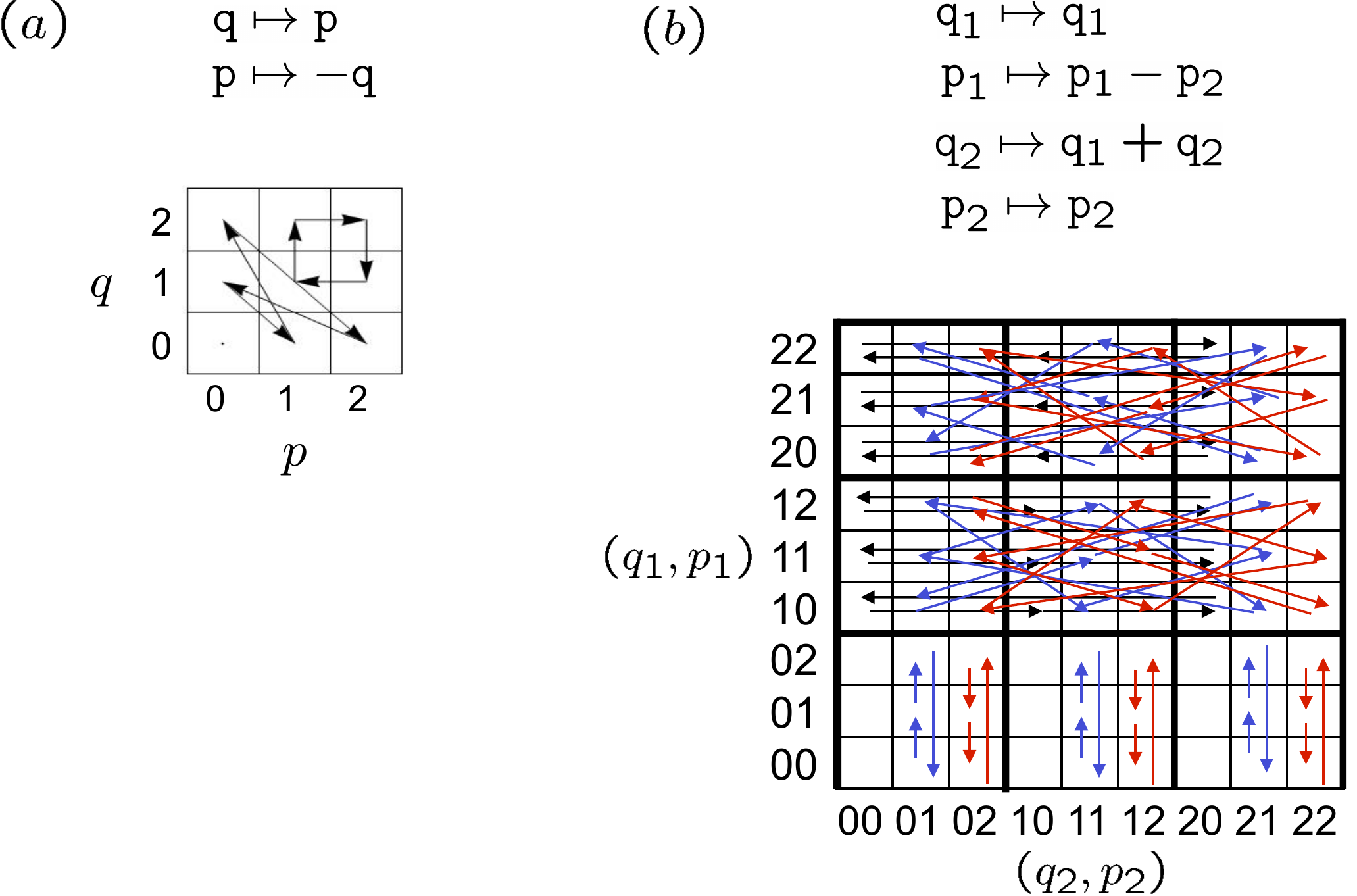} }
\caption{Examples of valid reversible transformations for (a) one trit, and (b) two trits.}
\label{fig:twotrittransf}
\end{figure}

The valid reversible transformations are the affine symplectic maps on the phase-space.  These correspond to a particular subset of the permutations.
 Some examples are depicted in Fig.~\ref{fig:twotrittransf}.

Just as in the continuous case, the valid measurements are those that determine the values of a set of Poisson-commuting quadrature variables.  For instance, for a single trit, there are only four inequivalent measurements: of $q$, of $p$, of $q+p$ and of $q+2p$, depicted in Fig.~\ref{fig:mmtstrits}(a), with different colours denoting different outcomes.  Fig.~\ref{fig:mmtstrits}(b) depicts some valid measurements on a pair of trits.  The left depicts a joint measurement of $q_1$ and $q_2$, which corresponds to a product basis in quantum theory.  The right depicts a joint measurement of $q_1-q_2$ and $p_1 + p_2$, which corresponds to a basis of entangled states.

\begin{figure}[h!]
 \center{   \includegraphics[width=8cm]{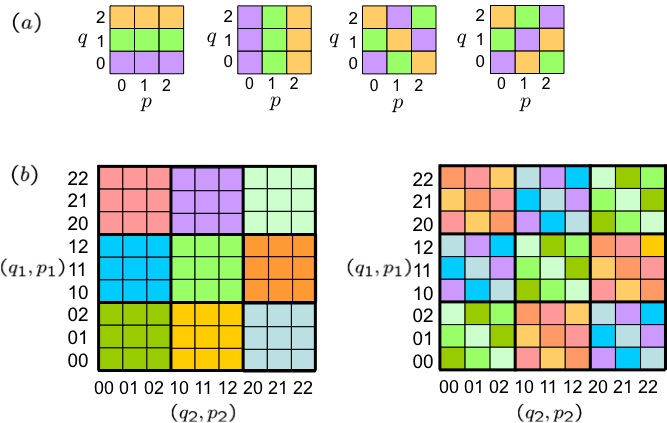}}
\caption{(a) The set of valid sharp measurements on a single trit.  (b) Examples of valid sharp measurements on a pair of trits.  The support of the response function corresponding to a particular outcome is coloured uniformly.}
\label{fig:mmtstrits}
\end{figure}

\subsection{Quadrature epistricted theory of bits}

The epistricted theory of bits is very similar to that of trits, except with $\mathbb{Z}_2$ rather than $\mathbb{Z}_3$ describing the configuration space of a single degree of freedom.
For a single system ($n=1)$, we can depict the phase space $\Omega$ as a 2$\times2$ grid.  There are only three inequivalent linear functionals:
\begin{equation}
q,\;p,\;q+p.
\end{equation}
Unlike the case of trits, $q-p$ is not a distinct functional because in arithmetic modulo 2, $q-p=q+p$.

It follows that the valid epistemic states for a single system are those depicted in Fig.~\ref{fig:epistemicstates1toybit}.  There are six pure states and one mixed state.  We adopt a similar graphical convention to depict the 4-dimensional phase space of a pair of bits as we did for a pair of trits, presented in Fig.~\ref{fig:4by4sudoku}.  Because the combinatorics are not so bad for the case of bits,  we depict \emph{all} of the valid epistemic states for a pair of bits in Fig.~\ref{fig:allepistemicstatestwobits}.  We categorize these into those for which two variables are known (the pure states) and those for which only one or no variable is known (the mixed states).  We also categorize these according to whether they exhibit correlation between the two subsystems or not.  The pure correlated states correspond to the entangled states.  

\begin{figure}[h!]
 \center{   \includegraphics[width=6cm]{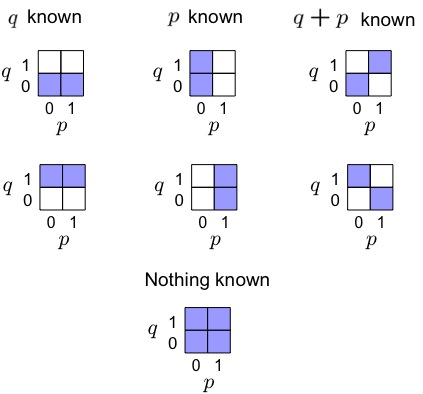}}
\caption{All the valid epistemic states for a single bit. There are six states of maximal knowledge (one variable known) and a single state of nonmaximal knowledge (no variable known).}
\label{fig:epistemicstates1toybit}
\end{figure}

\begin{figure}[h!]
 \center{   \includegraphics[width=6.5cm]{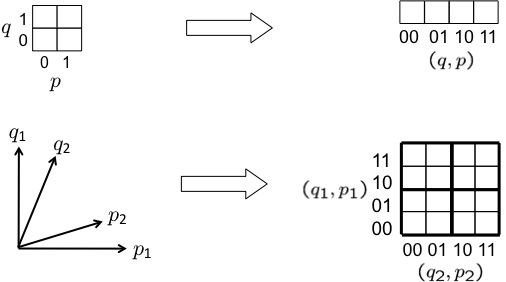}}
\caption{Embedding a 4-dimensional discrete phase space in a 2-dimensional grid.}
\label{fig:4by4sudoku}
\end{figure}

\begin{figure}[h!]
 \center{   \includegraphics[width=10cm]{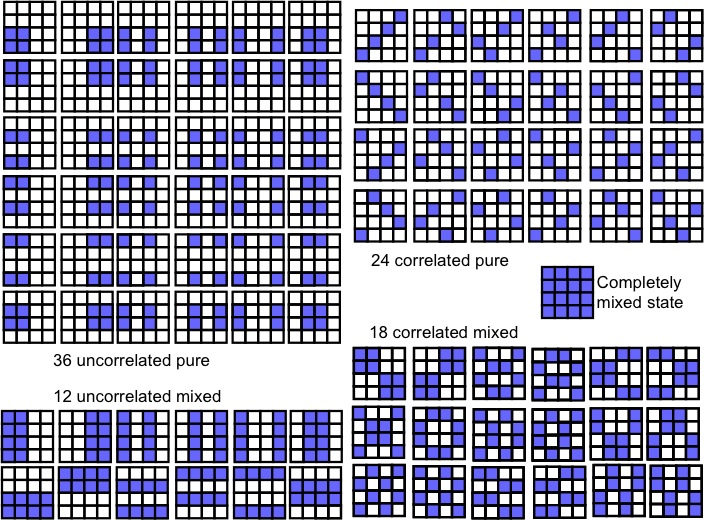}}
\caption{All the valid epistemic states for a pair of bits.}
\label{fig:allepistemicstatestwobits}
\end{figure}

The reversible transformations for the case of a single system $(n=1)$ are particularly simple.  In this case, $\Omega = (\mathbb{Z}_d)^2$, and the symplectic form is simply 
$J= 
\begin{pmatrix}
0 & 1 \\
1 & 0 
\end{pmatrix}.$
because $-1=1$ in arithmetic modulo 2.  As such, the symplectic matrices in this case are those with elements in $\mathbb{Z}_2$ and satisfying $S^T
\begin{pmatrix}
0 & 1 \\
1 & 0 
\end{pmatrix}
 S = \begin{pmatrix}
0 & 1 \\
1 & 0 
\end{pmatrix}$.  These are all the $2\times 2$ matrices having at least one column containing a $0$, that is,
\begin{equation}
\begin{pmatrix}
1 & 0 \\
0 & 1 
\end{pmatrix},\;\;
\begin{pmatrix}
0 & 1 \\
1 & 0 
\end{pmatrix},\;\;
\begin{pmatrix}
1 & 0 \\
1 & 1 
\end{pmatrix},\;\;
\begin{pmatrix}
1 & 1 \\
0 & 1 
\end{pmatrix},\;\;
\begin{pmatrix}
0 & 1 \\
1 & 1 
\end{pmatrix},\;\;
\begin{pmatrix}
1 & 1 \\
1 & 0 
\end{pmatrix},
\end{equation}
corresponding respectively to the tranformations
\begin{equation}
\begin{matrix}
{\tt q \mapsto q} \\
{\tt p \mapsto p}
\end{matrix}\;\;,\;\;
\begin{matrix}
{\tt q \mapsto p} \\
{\tt p \mapsto q}
\end{matrix}\;\;,\;\;
\begin{matrix}
{\tt q \mapsto q} \\
{\tt p \mapsto q+p}
\end{matrix}\;\;,\;\;
\begin{matrix}
{\tt q \mapsto q+p} \\
{\tt p \mapsto p}
\end{matrix}\;\;,\;\;
\begin{matrix}
{\tt q \mapsto p} \\
{\tt p \mapsto q+p}
\end{matrix}\;\;,\;\;
\begin{matrix}
{\tt q \mapsto q+p} \\
{\tt p \mapsto q}
\end{matrix}.
\end{equation}
Each of these symplectic transformations can be composed with the four possible phase-space displacements, 
\begin{equation}
\begin{matrix}
{\tt q \mapsto q} \\
{\tt p \mapsto p}
\end{matrix}\;\;,\;\;
\begin{matrix}
{\tt q \mapsto q+1} \\
{\tt p \mapsto p}
\end{matrix}\;\;,\;\;
\begin{matrix}
{\tt q \mapsto q} \\
{\tt p \mapsto p+1}
\end{matrix}\;\;,\;\;
\begin{matrix}
{\tt q \mapsto q+1} \\
{\tt p \mapsto p+1}
\end{matrix}.
\end{equation}
In all, this leads to 24 reversible symplectic affine transformations, which are depicted in Fig.~\ref{fig:transf1toybit}.  Given that there are only 24 permutations on the discrete phase space, we see that {\em every} reversible ontic transformation is physically allowed in this case.
\begin{figure}[h!]
 \center{   \includegraphics[width=10cm]{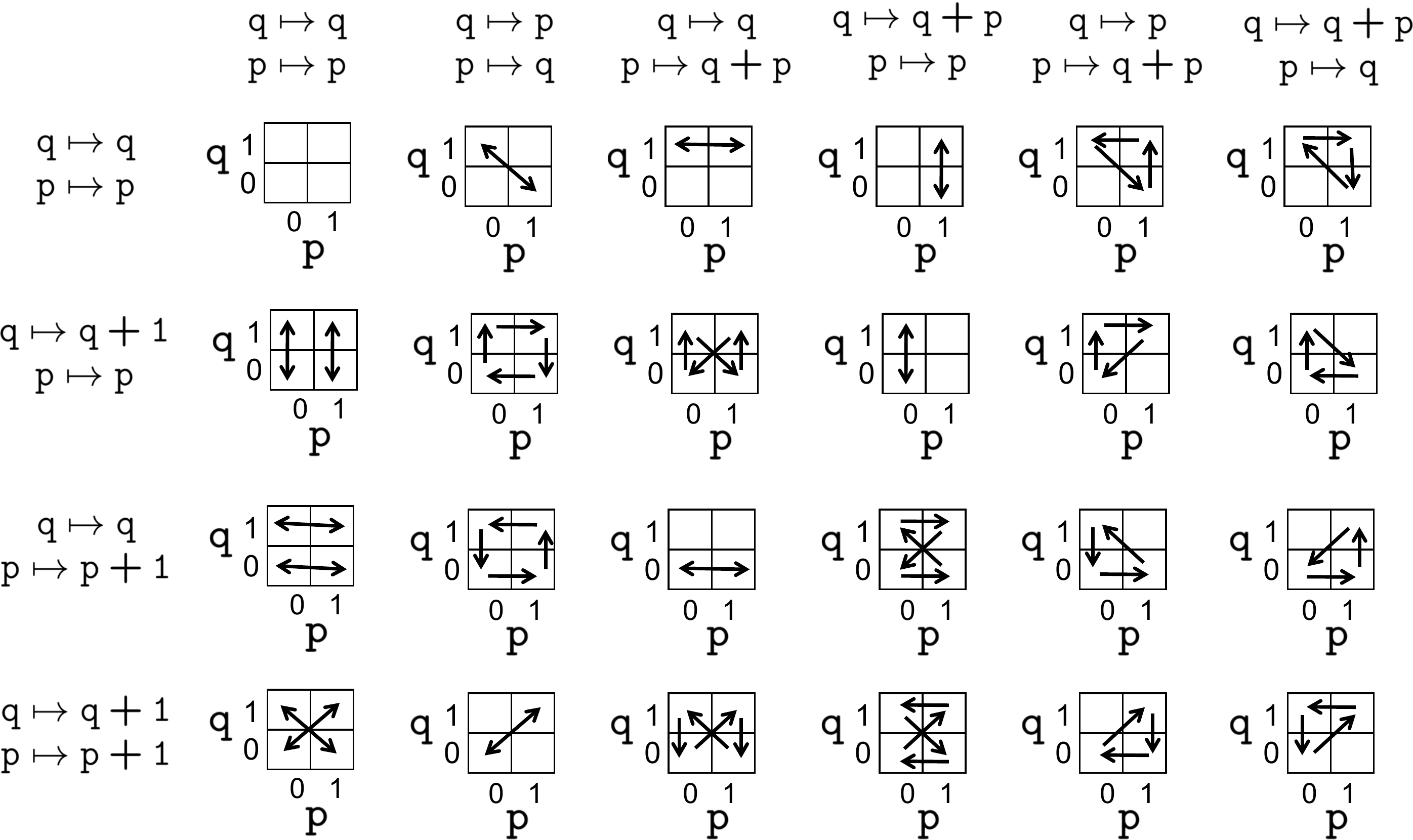}}
\caption{All the valid transformations for a single bit.}
\label{fig:transf1toybit}
\end{figure}
 

On the other hand, for a pair of systems $(n=2)$, only a subset of the permutations of the ontic states correspond to valid sympectic affine transformations. 

In Fig.~\ref{fig:mmtsbits}(a), we present the valid reproducible measurements on a single bit, and in Fig.~\ref{fig:mmtsbits}(b) we present some examples of such measurements on a pair of bits, one corresponding to a product basis and the other an entangled basis.

\begin{figure}[h!]
 \center{   \includegraphics[width=5cm]{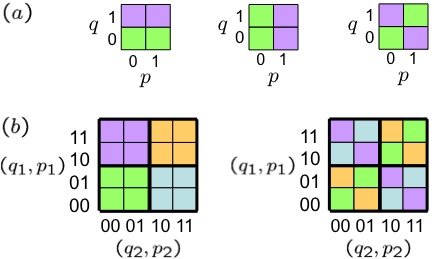}}
\caption{(a) The set of valid measurements on a single bit.  (b) Examples of valid measurements on a pair of bits.}
\label{fig:mmtsbits}
\end{figure}

\section{Quadrature quantum subtheories}\label{quadraturesubtheories}

We now shift our attention to quantum theory, and build up to a definition of the subtheories of quantum theory that our epistricted theories will ultimately be shown to reproduce.

\subsection{Quadrature observables}

We are interested in describing collections of elementary systems that each describe some continuous or discrete degree of freedom.   If the elementary system is a continuous degree of freedom, it is associated with the Hilbert space $\mathcal{L}^2(\mathbb{R})$, the space of square-integrable functions on $\mathbb{R}$. For the case of $n$ such systems, the Hilbert space is $\mathcal{L}^2(\mathbb{R})^{\otimes n}=\mathcal{L}^2(\mathbb{R}^n)$.  The sorts of discrete degrees of freedom we consider are those wherein all the elementary systems have $d$ levels where $d$ is a prime.  These are described by the Hilbert space $\mathbb{C}^d$. For $n$ such systems, the Hilbert space is $(\mathbb{C}^{d})^{\otimes n} = \mathbb{C}^{dn}$.

We seek to describe both discrete and continuous systems in the language of symplectic structure.  For a scalar field, for instance, we describe each mode of the field in terms of a pair of field quadratures.  In the example of a 2-level system, even though the physical degree of freedom in question may be spin or polarization, we seek to understand it in terms of a configuration variable and its canonically conjugate momentum.  In all of these cases, we will conventionally refer to the pair of conjugate variables, regardless of the degrees of freedom they describe, as `position' and `momentum'. 

We wish to present the quadrature subtheories for the continuous and discrete cases in a unified manner.  Towards this end, we will avoid using a Hermitian operator to represent the quantum measurement associated to a quadrature variable.  The reason is that although this works well for the continuous case, it fails to make sense in the discrete case.  Recall that in the continuous case, we can define Hermitian operators on $\mathcal{L}^2(\mathbb{R})$, denoted $\hat{q}$ and $\hat{p}$, and satisfying the commutation relation 
$[ \hat{q},\hat{p} ]= \hat{\mathbb{1}},$
where $[ \cdot, \cdot]$ denotes the matrix commutator and $\hat{\mathbb{1}}$ is the identity operator on $\mathcal{L}^2(\mathbb{R})$.  
In the discrete setting, however, we would expect the operators associated to the discrete position and momentum variables to have eigenvalues in the finite field $\mathbb{Z}_d$, whereas the eigenvalues of Hermitian operators are necessarily real.  
Even if we did pick a pair of Hermitian operators to serve as discrete position and momentum observables, these would necessarily fail to provide an analogue of the commutation relation $[ \hat{q},\hat{p} ]= \hat{\mathbb{1}},$
 because in a finite-dimensional Hilbert space, the commutator of any two Hermitian operators has vanishing trace and therefore cannot be proportional to the identity operator on that space.

In any case, within the fields of quantum foundations and quantum information, there has been a move away from representing measurements 
by Hermitian operators because  the eigenvalues of these operators are merely arbitrary labels of the measurement outcomes and have no operational significance. It is only the projectors in the spectral resolution of such a Hermitian operator that appear in the Born rule and hence only these that are relevant to the operational statistics.  Therefore, a measurement with outcome set $K$ is associated with a set of projectors $\{ \Pi_k : k \in K\}$ such that $\Pi_k^2 = \Pi_k,\; \forall k\in K$ and $\sum_{k\in K} \Pi_k = \mathbb{1}$ (integral in the case of a continuum of outcomes).
Such a set is called a \emph{projector-valued measure} (PVM).  

In the continuous variable case, we define the position observable, denoted $\mathcal{O}_q$, to be the PVM consisting of projectors onto position eigenstates,
\[
\mathcal{O}_q \equiv \{ \hat{\Pi}_q( {\tt q})  : {\tt q}\in \mathbb{R} \},
\]
where
\[
\hat{\Pi}_q( {\tt q}) \equiv |{\tt q}\rangle_q \langle {\tt q}|.
\]
The momentum observable, denoted $\mathcal{O}_p$, is defined to be the PVM of projectors onto momentum eigenstates
\[
\mathcal{O}_p \equiv \{\hat{\Pi}_p( {\tt p}) : {\tt p} \in \mathbb{R} \},
\]
where
\[
\hat{\Pi}_p( {\tt p}) \equiv |{\tt p}\rangle_p \langle {\tt p}|,
\]
and the momentum eigenstates 
are related to the position eigenstates by a Fourier transform,
\begin{equation}
	| {\tt p} \rangle_p \equiv \frac{1}{2 \pi \hbar} \int_{\mathbb{R}} {\rm d}{\tt q} e^{i \frac{ {\tt qp}}{\hbar}} | {\tt q}\rangle_q.
\end{equation}
Strictly speaking, one needs to make use of rigged Hilbert space to define position and momentum eigenstates rigorously but we will adopt the standard informal treatment of such states here.


In the discrete case, we can also define position and momentum observables in this way.  A discrete position basis for $\mathbb{C}^d$  (which one can think of as the {\it computational basis} in a quantum information setting) can be chosen arbitrarily.  Denoting this basis by  $\{ |{\tt q}\rangle_q : {\tt q}\in  \mathbb{Z}_d \}$, the PVM defining the position observable, denoted $\mathcal{O}_q$, is
\[
\mathcal{O}_q \equiv \{  \hat{\Pi}_q( {\tt q}) : {\tt q}\in \mathbb{Z}_d \},
\]
where $\hat{\Pi}_q( {\tt q}) \equiv |{\tt q}\rangle_q \langle {\tt q}|$. We can define a discrete momentum basis, denoted $\{ |{\tt p}\rangle_p : {\tt p} \in \mathbb{Z}_d \}$, via a discrete Fourier transform,
\begin{equation}
	| {\tt p}\rangle_p \equiv \frac{1}{\sqrt {d}} \sum_{{\tt q} \in \mathbb{Z}_d} e^{i 2\pi \frac{ {\tt qp}}{d}} | {\tt q}\rangle.
\end{equation}
and in terms of it, the PVM defining the momentum observable,
\[
\mathcal{O}_p \equiv \{ \hat{\Pi}_p( {\tt p}) : {\tt p} \in  \mathbb{Z}_d \},
\]
where $\hat{\Pi}_p( {\tt p}) \equiv |{\tt p}\rangle_p \langle {\tt p}|$.
If one does not associate a Hermitian operator to each observable, then joint measurability of two observables can no longer be decided by the commutation of the associated Hermitian operators.  Rather, it is determined by whether the associated PVMs commute or not, where two PVMs are said to commute if every projector in one commutes with every projector in the other. 


To define the rest of the quadrature observables (and the commuting sets of these), we must first define a unitary representation of the symplectic affine transformations.
We begin by specifying the unitaries that correspond to phase-space displacements.  
To do this in a uniform manner for discrete and continuous degrees of freedom, we define functions $\chi: \mathbb{R} \to \mathbb{C}$ and $\chi: \mathbb{Z}_d \to \mathbb{C}$ as
\begin{eqnarray}
\chi({\tt c}  ) &=& e^{i \frac{{\tt c}}{\hbar}} \textrm{ for } {\tt c}\in \mathbb{R}\nonumber \\
\chi({\tt c}  ) &=& e^{i \frac{2\pi}{d} {\tt c}} \textrm{ for } {\tt c}\in \mathbb{Z}_d, \textrm{ when $d$ is an odd prime}\nonumber \\
\chi({\tt c} ) &=& e^{i \frac{\pi}{2} {\tt c}} \textrm{ for } {\tt c}\in \mathbb{Z}_d, \textrm{ when $d=2$.}
\end{eqnarray}
In the continuous case, this is the standard exponential function; in the discrete case where $d$ is an odd prime, $\chi({\tt a})$ is the ${\tt a}$th power of the  $d$th root of unity; in the discrete case where $d=2$, $\chi({\tt a})$ is the ${\tt a}$th power of the  fourth (not the second) root of unity.
In terms of this function, we can define a unitary that shifts the position by ${\tt q}$, where ${\tt q} \in \mathbb{R}$ in the continuous case and ${\tt q} \in \mathbb{Z}_d$ in the discrete case, as
\begin{align}
\hat{S}({\tt q}) &= \sum_{{\tt p} \in \mathbb{R}/\mathbb{Z}_d} \chi({\tt q p}) |{\tt p}\rangle_{p} \langle {\tt p}|\nonumber\\
&= \sum_{{\tt q}' \in \mathbb{R}/\mathbb{Z}_d} |{\tt q}'-{\tt q}\rangle_q \langle {\tt q}'|
\end{align}
and a unitary that boosts the momentum by ${\tt p}$, where ${\tt p} \in \mathbb{R}$ in the continuous case and ${\tt p} \in \mathbb{Z}_d$ in the discrete case, as
\begin{align}
\hat{B}({\tt p}) &= \sum_{{\tt q} \in \mathbb{R}/\mathbb{Z}_d} \chi({\tt q p}) |{\tt q}\rangle_q \langle {\tt q}|\nonumber\\
&= \sum_{{\tt p}' \in \mathbb{R}/\mathbb{Z}_d} |{\tt p}'-{\tt p}\rangle_{p} \langle {\tt p}'|
\end{align}
Note that the shift unitaries do not commute with the boost unitaries.  The unitaries corresponding to phase-space displacements---typically called the \emph{Weyl operators}---are proportional to products of these. In particular, the Weyl operator associated with the phase-space displacement vector ${\bf a} = ({\tt q},{\tt p}) \in \mathbb{R}^2/(\mathbb{Z}_d)^2$ is defined to be
\begin{equation} \label{def:Weyl}
{\hat W}({\bf a}) = \chi(2 {\tt pq}) \hat{S}({\tt q}) \hat{B}({\tt p}).
\end{equation}
This is easily generalized to the case of a phase-space displacement for $n$ degrees of freedom,
${\bf a} = ({\tt q}_1,{\tt p}_1,\dots,{\tt q}_n,{\tt p}_n) \in \mathbb{R}^{2n}/(\mathbb{Z}_d)^{2n}$ via the tensor product,
\begin{equation} \label{generalWeyl}
{\hat W}({\bf a}) = \bigotimes_{i=1}^n \chi(2 {\tt p}_i {\tt q}_i) \hat{S}({\tt q}_i) \hat{B}({\tt p}_i).
\end{equation}
For ${\bf a},{\bf a}' \in \Omega$, the product of the corresponding Weyl operators is
\begin{equation} \label{productWeyl}
{\hat W}({\bf a}){\hat W}({\bf a}')= \chi(2 \langle {\bf a}, {\bf a}'\rangle) {\hat W} ( {\bf a} + {\bf a}').
\end{equation}
Thus it is clear that the Weyl operators constitute a projective unitary representation of the group of phase-space displacements ${\bf m}\to {\bf m}+{\bf a}$, where the composition law is
\begin{equation}\label{phasespacedispls}
(\cdot + {\bf a}) + {\bf a}' = \cdot + ({\bf a} +{\bf a}').
\end{equation}

Next, we define a projective unitary representation $\hat{V}$ of the symplectic group acting on a $2n$-dimensional phase space $\Omega$.  For every $2n \times 2n$ symplectic matrix $S : \Omega \to \Omega$, there is a unitary $\hat{V}(S)$ acting on the Hilbert space $\mathcal{L}(\mathbb{R}^n)/\mathbb{C}^{dn}$, such that
\begin{equation}\label{productsymp}
\hat{V}(S)\hat{V}(S')=e^{i\phi} \hat{V}(SS')
\end{equation}
for some phase factor $e^{i \phi}$.  These can be defined via their action on the Weyl operators.  Specifically, $\forall {\bf a} \in \Omega$,
\begin{equation}\label{SympactiononWeyl}
\hat{V}(S) \hat{W}({\bf a}) \hat{V}^{\dag}(S) \propto \hat{W}(S {\bf a}).
\end{equation}

 In the following, we will often consider the action of these unitaries under conjugation, therefore, we define the superoperators associated to phase-space displacement ${\bf a}$ and symplectic matrix $S$,
\begin{align}
\mathcal{W}({\bf a}) (\cdot) &\equiv \hat{W}({\bf a}) \cdot \hat{W}({\bf a})^{\dag},\nonumber \\
\mathcal{V}(S) (\cdot) &\equiv \hat{V}(S) \cdot \hat{V}(S)^{\dag}.
\end{align}
Note that Eq.~\eqref{SympactiononWeyl} implies that
\begin{align}
\mathcal{W}({\bf a}) \circ \mathcal{V}(S) (\cdot) =\mathcal{V}(S) \circ \mathcal{W}(S^{-1}{\bf a}) (\cdot).
\end{align}

In the classical theory, every Poisson-commuting set of quadrature functionals
$\{ f^{(1)}, f^{(2)}, \dots, f^{(k)} \}$ 
can be obtained from every other such set by a symplectic linear transformation (here, $k\le n$).  The proof is as follows.  If $f^{(i)}={\bf f^{(i)}}^T{\bf z}$ is a quadrature functional, then so is $\tilde{f}^{(i)}= (S{\bf f}^{(i)})^T {\bf z}$ for all $i \in \{ 1, \dots ,k\}$ when $S$ is a symplectic matrix.  Furthermore, if the initial set is Poisson-commuting, then $\langle {\bf f}^{(i)}, {\bf f}^{(j)} \rangle =0$ for all $i \ne j \in \{ 1, \dots, k\}$, and then because
\begin{align}
\langle {\bf \tilde{f}}^{(i)}, {\bf \tilde{f}}^{(j)} \rangle &= \langle S{\bf f}^{(i)}, S{\bf f}^{(j)} \rangle \nonumber\\
 &= \langle {\bf f}^{(i)}, {\bf f}^{(j)} \rangle,
\end{align}
it follows that $\langle {\bf \tilde{f}}^{(i)}, {\bf \tilde{f}}^{(j)} \rangle =0$ for all $i \ne j \in \{ 1, \dots, k\}$ so the final set is Poisson-commuting as well.  Here, we have used the fact that the symplectic inner product is invariant under the action of a symplectic matrix.

We can define commuting sets of quantum quadrature \emph{observables} similarly.   
Consider a single degree of freedom, $\Omega = \mathbb{R}^2/\mathbb{Z}_d^2$.  Denote by $S_f$ the symplectic matrix that takes 
the position functional $q$ to a 
quadrature functional $f$, so that $S_f {\bf q} = {\bf f}$.  (Given that ${\bf q} \equiv (1,0)$, we see that ${\bf f}$ is the first column of $S_f$.) We define the quadrature \emph{observable} associated with $f$, denoted $\mathcal{O}_f$,
to be the image under the action of the unitary $\hat{V}(S_f)$ of the position observable, that is,
\[
\mathcal{O}_f \equiv \{ \hat{\Pi}_f ({\tt f}): {\tt f} \in \mathbb{Z}_d \},
\]
where
\begin{align}
 \hat{\Pi}_f ({\tt f}) &\equiv   \mathcal{V}(S_f) (\hat{\Pi}_q ({\tt f}) ).
\end{align}

It is useful to note how these projectors transform under phase-space displacements and symplectic matrices. 
By definition of the quadrature observables, we infer that  for a symplectic matrix $S$,
\begin{equation} \label{effectsymp}
\mathcal{V}(S) (\hat{\Pi}_{f}({\tt f}) ) = \hat{\Pi}_{Sf}({\tt f}),
\end{equation}
where $S f$ denotes the quadrature functional associated to the vector $S {\bf f}\in \Omega$.  
Now consider the action of a Weyl superoperator.  First note that the projectors onto position eigenstates transform as
\[
\mathcal{W}({\bf a}) (\hat{\Pi}_q({\tt q}))=\hat{\Pi}_q({\tt q}+q({\bf a})).
\]
It follows that if $f = S_f q$, then
\begin{align} \label{effectaffine}
\mathcal{W}({\bf a}) (\hat{\Pi}_f({\tt f}))&=\mathcal{W}({\bf a}) ( \hat{\Pi}_{S_f q}({\tt f}) ),\nonumber\\
&=\mathcal{W}({\bf a}) \mathcal{V}(S_f ) (\hat{\Pi}_q ({\tt f}) ),\nonumber\\
&=\mathcal{V}(S)  \mathcal{W}(S_f^{-1}{\bf a}) (\hat{\Pi}_q ({\tt f}) ),\nonumber \\
&=\mathcal{V}(S)  (\hat{\Pi}_q ({\tt f}+ q(S_f^{-1}{\bf a})) ),\nonumber \\
&=\hat{\Pi}_f({\tt f}+f({\bf a})).
\end{align}

In all, 
\begin{align} \label{effectaffine}
\mathcal{V}(S) \mathcal{W}({\bf a}) (\hat{\Pi}_f({\tt f}))&= \hat{\Pi}_{S f}({\tt f} + f({\bf a})).
\end{align}




The case of $n$ degrees of freedom, $\Omega = \mathbb{R}^2/\mathbb{Z}_d^2$, is treated similarly.   In this case, our quadrature observables need not be rank-1.  Our fiducial quadrature can be taken to be $q_1$, the position functional for system 1.  The associated quadrature observable is
\[
\mathcal{O}_{q_1} \equiv \{  \hat{\Pi}_{q_1} ({\tt q}_1)  \otimes \mathbb{1}_2 \otimes \cdots \otimes \mathbb{1}_n : {\tt q}_1 \in \mathbb{R}/\mathbb{Z}_d \}.
\]
 For an arbitrary functional on the $n$ systems, $f:\Omega \to \mathbb{R}/\mathbb{Z}_d$, we find the symplectic matrix $S_f$ such that $S_f {\bf q}_1 = {\bf f}$, and we define the quadrature observable associated with $f$ to be
\[
\mathcal{O}_{f} \equiv \{  \hat{\Pi}_{f}({\tt f}) : {\tt f} \in \mathbb{R}/\mathbb{Z}_d \}.
\]
where
\[
\hat{\Pi}_{f}({\tt f}) \equiv \hat{V}(S_f) \left( \hat{\Pi}_{q_1} ({\tt f}) \otimes \mathbb{1}_2 \otimes \cdots \otimes \mathbb{1}_n \right) \hat{V}(S_f)^{\dag}.
\]
It follows that for every classical quadrature \emph{functional} $f$, there is a corresponding quadrature \emph{observable} $\mathcal{O}_f$, which stands in relation to the position and momentum observables as $f$ stands to the position and momentum functionals.

As an aside, one may note that in the continuous variable case, the quadrature observables are simply the spectral resolutions of those Hermitian operators that are linear combinations of position and momentum operators. 
In particular, for a quadrature observable $\mathcal{O}_f$ associated to a vector ${\bf f} \in \Omega$, the associated Hermitian operator is simply
\[
\hat{f} = {\bf f}^T {\bf \hat{z}},
\]
where
\begin{equation}\label{vectorqpobservables}
{\bf \hat{z}} \equiv (\hat{q}_1 ,\hat{p}_1, \dots,\hat{q}_n,\hat{p}_n),
\end{equation}
is the vector of position and momentum operators. 
Hence for every classical quadrature \emph{variable} $f = {\bf f}^T {\bf z}$, as defined in Eq.~\eqref{eq:f}, there is a corresponding quadrature \emph{operator} $\hat{f}={\bf f}^T {\bf \hat{z}}$, where we have simply replaced the position and momentum functionals with their corresponding Hermitian operators. 

We are now in a position to describe the commuting sets of quadrature observables.  A set of quadrature observables $\{ \mathcal{O}_{f^{(1)}}, \dots, \mathcal{O}_{f^{(k)}} \}$ is a commuting set if and only if the corresponding quadrature functionals $\{ f^{(1)}, \dots, f^{(k)} \}$ are Poisson-commuting.  The proof is as follows.  The functionals $\{ f^{(1)}, \dots, f^{(k)} \}$ are Poisson-commuting if and only if they can be obtained by some symplectic transformation from any other such set, in particular, the set of position functionals for the first $k$ systems,  $\{ q_1, \dots, q_k \}$.  In other words, $\{ f^{(1)}, \dots, f^{(k)} \}$ are Poisson-commuting if and only if there is a sympectic matrix $S$  such that ${\bf f}^{(i)} = S{\bf q}_i$ for all $i \in \{1,\dots, k\}$ (which implies that the vectors ${\bf f}^{(i)}$ are the first $k$ columns of $S$).
Given the definition of quadrature observables, this condition is equivalent to the statement that there exists a symplectic matrix $S$ such that $\mathcal{O}_{f^{(i)}} = \hat{V}(S) \mathcal{O}_{q_i} \hat{V}(S)^{\dag}$ for all $i \in \{1,\dots, k\}$.
But given that the  elements of the set $\{ \mathcal{O}_{q_1}, \dots, \mathcal{O}_{q_k} \}$ (the position observables for the first $k$ systems) commute, and commutation relations are preserved under a unitary, it follows that the elements of the set $\{ \mathcal{O}_{f^{(1)}}, \dots, \mathcal{O}_{f^{(k)}} \}$ commute if and only if there exists such an $S$, hence they commute if and only if the corresponding quadrature functionals $\{ f^{(1)}, \dots, f^{(k)} \}$ Poisson-commute. 

Again, this has a simple interpretation in the continuous variable case.  There,  it is easy to verify that the matrix commutator of two quadratures operators is equal to the symplectic inner product of the corresponding vectors, that is, 
$[\hat{f},\hat{g}]= \langle {\bf f},{\bf g}\rangle.$
In particular, it follows that $[\hat{f},\hat{g}]=0$ if and only if $\langle {\bf f},{\bf g}\rangle=0$,
which provides another proof of the fact that a commuting set of quadrature observables is associated with an isotropic subspace of the phase space.

As described in Sec.~\ref{validepistemicstates}, every set of Poisson-commuting quadrature functionals defines an istropic subspace $V\subseteq \Omega$ and therefore the sets of commuting quadrature observables are also parameterized by the isotropic subspaces of $\Omega$.  If a commuting set of quadrature observables is such that the corresponding quadrature functionals are associated with an isotropic subspace $V$, then this set defines a single quadrature observable, denoted $\mathcal{O}_V$, by
\[
\mathcal{O}_V = \{ \hat{\Pi}_{V}({\bf v}): {\bf v} \in V\}
\]
where
\begin{equation}\label{pivv}
\hat{\Pi}_{V}({\bf v}) \equiv \prod_{{\bf f}^{(i)}: {\rm span}({\bf f}^{(i)})=V} \hat{\Pi}_{f^{(i)}}\left(f^{(i)}({\bf v})\right).
\end{equation}
For instance, the quadrature functionals $f = q_1 - q_2$ and $g = p_1 +p_2$ are Poisson-commuting and therefore the associated quadrature observables, $\mathcal{O}_f$ and $\mathcal{O}_g$, commute, which is to say that the projectors $\{\hat{\Pi}_{f}({\tt f}) : {\tt f}\in \mathbb{R}/\mathbb{Z}_d \}$ all commute with the projectors $\{\hat{\Pi}_{g}({\tt g}) : {\tt g}\in \mathbb{R}/\mathbb{Z}_d \}$.  If $V = {\rm span}\{{\bf f},{\bf g}\}$, then the possible pairs of values for the two observables can be expressed as the possible components of a vector ${\bf v} \in V$ along the basis vectors ${\bf f}$ and ${\bf g}$ respectively.  
These are the pairs $\{ ({\tt f}, {\tt g})\}$ such that ${\tt f} = {\bf f}^T{\bf v} = f({\bf v}) $ and ${\tt g} = {\bf g}^T{\bf v} = g({\bf v}) $ for some ${\bf v} \in V$.
It follows that we can parametrize the possible values of this commuting set by vectors ${\bf v}\in V$.   

In the continuous variable case, $\hat{\Pi}_f({\tt f})$ is the projector onto the eigenspace of $\hat{q}_1 - \hat{q}_2$ with eigenvalue ${\tt f}$, $\hat{\Pi}_g({\tt g})$ is the projector onto the eigenspace of $\hat{p}_1 + \hat{p}_2$ with eigenvalue ${\tt g}$, and $\hat{\Pi}_{V}({\bf v})$ is the projector onto the joint eigenspace of $\hat{q}_1 - \hat{q}_2$ and $\hat{p}_1 + \hat{p}_2$ with eigenvalues $({\tt f},{\tt g})$, which corresponds to an Einstein-Podolsky-Rosen entangled state.

With this background established, we are in a position to define the quadrature quantum subtheories. 

\subsection{Characterization of quadrature quantum subtheories}

In this section, we define a quadrature subtheory of the quantum theory for a given system (discrete or continuous).
In the discrete case, this subtheory is closely connected to the stabilizer formalism, a connection that we make precise in the appendix.

\subsubsection{The set of valid quantum states}


In order to define the valid quantum states in the quadrature quantum subtheory, we use the guiding analogy of Sec.~\ref{complementarity}, together with the isomorphism between quadrature functionals and quadrature observables noted above.

As we have just seen, for both the discrete and continuous cases, every commuting set of quadrature observables is associated to an isotropic subspace $V \subset \Omega$.  Furthermore, every set of values that these observables can jointly take is associated to a vector ${\bf v} \in V$. We have denoted the projector that yields these values by $\hat{\Pi}_V({\bf v})$.  
The quantum states that are part of the quadrature subtheory, termed \emph{quadrature states}, 
are simply the density operators that are proportional to such projectors. 
It follows that the quadrature states are parameterized by pairs consisting of an isotropic subspace $V$ and a valuation vector ${\bf v} \in V$ (in precisely the same way as one parametrizes the set of valid epistemic states in the epistricted classical theory).  Specifically, it is the set of states of the form
\begin{equation}\label{qdelt}
\rho_{V,{\bf v}} = \frac{1}{\mathcal{N}_V } \hat{\Pi}_{V}({\bf v}),
\end{equation}
where $V\subseteq \Omega$ is isotropic, ${\bf v} \in V$, and $\mathcal{N}_V$ is a normalization factor.
Equivalently, if $\{ {\bf f}^{(i)} \}$ is a basis of $V$, then 
\begin{equation}\label{qprod}
\rho_{V,{\bf v}} \equiv \frac{1}{\mathcal{N}_V } \prod_{\{ {\bf f}^{(i)} : {\rm span}\{ {\bf f}^{(i)} \} =V \} } \hat{\Pi}_{f^{(i)}} \left( f^{(i)}({\bf v}) \right).
\end{equation}

\subsubsection{The set of valid transformations}

Because the overall phase of a Hilbert space vector is physically irrelevant, physical states are properly represented by density operators, and consequently a reversible physical transformation is not represented by a unitary operator but rather by the superoperator corresponding to conjugation by that unitary.  

When a Weyl operator $\hat{W}({\bf a})$ acts by conjugation, it defines what we will call the \emph{Weyl superoperator}, 
\[
\mathcal{W}({\bf a})(\cdot) \equiv \hat{W}({\bf a}) (\cdot) \hat{W}({\bf a})^{\dag}.
\]
Unlike the Weyl operators of two phase-space displacements, which,
by Eq.~\eqref{productWeyl}, commute if and only if the corresponding phase-space displacement vectors have vanishing symplectic inner product, 
\[
[ \hat{W}({\bf a}) , \hat{W}({\bf a}') ]=0 \;\textrm{  if and only if }\;   \langle {\bf a}, {\bf a}'\rangle=0,
\]
the Weyl superoperators of any two phase-space displacements necessarily commute,
\[
[\mathcal{W}({\bf a}) ,\mathcal{W}({\bf a}') ]=0\; \forall {\bf a},{\bf a}' \in \Omega.
\]
This follows from Eq.~\eqref{productWeyl} and the skew-symmetry of the symplectic inner product.  It follows that
\[
\mathcal{W}({\bf a}) \mathcal{W}({\bf a}')  = \mathcal{W}({\bf a} + {\bf a}')\; \forall {\bf a},{\bf a}' \in \Omega.
\]
 As such, the Weyl superoperators constitute a nonprojective representation of the group of phase-space displacements, Eq.~\eqref{phasespacedispls}.

Next, we consider the projective unitary representation $\hat{V}$ of the symplectic group acting by conjugation. 
This defines a superoperator representation of the symplectic group which is nonprojective, that is, for 
\[
\mathcal{V}(S)(\cdot) \equiv \hat{V}(S) (\cdot) \hat{V}(S)^{\dag},
\]
we have
\[
\mathcal{V}(S)\mathcal{V}(S')=\mathcal{V}(SS').
\]

The Clifford group of unitaries is defined as those which, when acting by conjugation, take the set of Weyl operators to itself.   That is, a unitary $\hat{U}$ is in the Clifford group if $\forall {\bf b} \in \Omega$,
\begin{equation}\label{CliffordactiononWeyl}
\hat{U} \hat{W}({\bf b}) \hat{U}^{\dag} = c({\bf b}) \hat{W}(S {\bf b}),
\end{equation}
for some maps $c : \Omega \to \mathbb{C}$ and $S:  \Omega \to\Omega$.

It turns out that every such unitary can be written as a product of a Weyl operator and an element of the unitary projective representation of the symplectic group, that is,
\[
\hat{U}(S,{\bf a}) = \hat{W}({\bf a}) \hat{V}(S),
\]
for some symplectic matrix $S: \Omega \to \Omega$ and phase-space vector ${\bf a} \in \Omega$.  
From Eqs.~\eqref{productWeyl} and \eqref{productsymp} we infer that a product of such unitaries is 
\begin{equation}\label{Ucomposition}
\hat{U}(S,{\bf a})\hat{U}(S',{\bf a}')= e^{i\phi} \hat{U}(SS', S{\bf a}' + {\bf a} ).
\end{equation}
for some phase factor $\phi$.  Recalling Eq.~\eqref{sympaffine}, it is clear that the Clifford group of unitaries $\hat{U}(S,{\bf a})$ constitutes a projective representation of the symplectic affine group.

When a Clifford unitary $\hat{U}(S,{\bf a})$ acts by conjugation, it defines what we will call a \emph{Clifford superoperator} $\mathcal{U}(S,{\bf a})(\cdot) \equiv \hat{U}(S,{\bf a}) (\cdot) \hat{U}(S,{\bf a})^{\dag}$. 
It follows that 
\[
\mathcal{U}(S,{\bf b})\mathcal{U}(S',{\bf b}')= \mathcal{U}(SS', {\bf b} + S{\bf b}'),
\]
and therefore, recalling Eq.~\eqref{sympaffine}, these form a nonprojective representation of the symplectic affine group.

The reversible transformations that are included in quadrature quantum mechanics are precisely those associated with Clifford superoperators.  These map every quadrature state to another quadrature state.

The valid \emph{irreversible} transformations in the quadrature subtheory are those that admit of a Stinespring dilation of the following form: the system is coupled to an ancilla of arbitrary dimension that is prepared in a quadrature state, the system and ancilla undergo a reversible transformation associated with a Clifford superoperator, 
and a partial trace operation is performed on the ancilla.  

\subsubsection{The set of valid measurements}

Finally, the reproducible measurements included in quadrature quantum mechanics are simply those associated with a commuting set of quadrature observables.  Recall that these are parametrized by the isotropic subspaces $V \subset \Omega$, and correspond to PVMs of the form $\{ \hat{\Pi}_V({\bf v}) : {\bf v}\in V\}$, as defined in Eq.~\eqref{pivv}. 

The most general measurement allowed is one whose Naimark extension can be achieved by preparing an ancilla in a quadrature state, coupling to the system via a Clifford superoperator, and finally measuring a commuting set of quadrature observables on the ancilla.  

\section{Comparing quantum subtheories to epistricted theories}

\subsection{Equivalence for continuous and odd-prime discrete cases}

The operational equivalence result is proven using the Wigner representation.  The latter is a quasi-probability representation of quantum mechanics, wherein Hermitian operators on the Hilbert space are represented by real-valued functions on the corresponding classical phase space.  

For the case of $n$ continuous degrees of freedom, where the Hilbert space is $\mathcal{L}^2(\mathbb{R}^{n})$ and the phase space is $\mathbb{R}^{2n}$, the Wigner representation is a well-known formulation of quantum theory, particularly in the field of quantum optics~\cite{wigner1932quantum,gardiner2004quantum}.  For the case of \emph{discrete} degrees of freedom, there are many proposals for how to define a quasi-probability representation that is analogous to Wigner's but for a discrete phase space.  We here make use of a proposal due to Gross~\cite{gross2006hudson}, which is built on (but distinct from) a proposal  by Wootters~\cite{gibbons2004discrete}.  For $n$ $d$-level systems (qudits), where $d$ is a prime, the phase space is taken to be $(\mathbb{Z}_d)^{2n}$.

We shall attempt to present the proof for the continuous case and for the odd-prime discrete case in a unified notation.  Towards this end, we will provide a definition of the Wigner representation that is independent of the nature of the phase space. In the case of $\Omega = \mathbb{R}^{2n}$ and $\Omega = (\mathbb{Z}_d)^{2n}$ for $d$ an odd prime, our definition will reduce, respectively, to the standard Wigner representation and the discrete Wigner representation proposed by Gross~\cite{gross2006hudson}.   Marginalizing over the entire phase space $\Omega$ will be denoted by a sum over $\Omega$ in all of our expressions, which will be taken to represent a discrete sum in the discrete case and an integral with a phase-space invariant measure in the continuous case.  


\subsubsection{Wigner representation of quantum theory}

The Wigner representation of an operator $\hat{O}$, denoted $\hat{W}_{\hat{O}}({\bf m})$, can be understood as the components of that operator in a particular basis  for the vector space of Hermitian operators where the inner product is the Hilbert-Schmidt inner product, $\langle \hat{O}, \hat{O}' \rangle \equiv {\rm tr}(\hat{O} \hat{O}')$.  The elements of this operator basis are indexed by the elements of the phase space and termed the \emph{phase-space point operators}.  
Denoting this operator basis by $\{ \hat{A}({\bf m}) : {\bf m}\in \Omega\}$, we have
\begin{equation}
  \hat{W}_{\hat{O}}({\bf m})=\mathrm{Tr}[\hat{O} \hat{A}({\bf m})]\,.
\end{equation}

The phase-space point operators can be defined as the symplectic Fourier transform of the Weyl operators (which in turn are defined for both continuous and discrete degrees of freedom in Eq.~\eqref{generalWeyl}),
  \begin{equation}\label{defnpointoperators}
\hat{A}({\bf m}) \equiv \frac{1}{\mathcal{N}_{\Omega}} \sum_{{\bf m}' \in \Omega} \chi(\langle {\bf m}, {\bf m}' \rangle) \hat{W}({\bf m}').
\end{equation}
where $\mathcal{N}_{\Omega}$ is a normalization factor chosen to ensure that 
\[
{\rm Tr}[\hat{A}({\bf m})]=1.
\]
The key property of the phase-space point operators is that they transform covariantly under symplectic affine transformations, 
\begin{equation}\label{pointcovariance}
\mathcal{U}(S,{\bf a}) \left[ \hat{A}({\bf m}) \right] \propto \hat{A}(S{\bf m}+{\bf a}),
\end{equation}
which can be inferred from Eq.~\eqref{defnpointoperators} and the manner in which the Weyl operators transform under the action of the Clifford superoperators, Eq.~\eqref{SympactiononWeyl}.  This in turn implies that the Wigner representation of an operator also transforms covariantly under symmplectic affine transformations,
\begin{align}
  \hat{W}_{\mathcal{U}(S,{\bf a})(\hat{O})}({\bf m})&= {\rm tr}\left( \mathcal{U}(S,{\bf a})(\hat{O}) \hat{A}({\bf m})   \right)\nonumber\\
  &= {\rm tr}\left(\hat{O}\;  \mathcal{U}(S^{-1}, -{\bf a}) ( \hat{A}({\bf m}) )  \right)\nonumber \\
  &= \hat{W}_{\hat{O}}(S^{-1}{\bf m}-{\bf a}).
\end{align}

In both the discrete and continuous cases, we have
\[
\frac{1}{\mathcal{N}_{\Omega}}  \sum_{{\bf m}\in \Omega} \chi(\langle {\bf m}, {\bf m}' \rangle) = \delta_{\bf 0}({\bf m}'),
\]
where $\delta_{\bf 0}({\bf m}') = \prod_{i=1}^n \delta({\tt q}'_i)\delta({\tt p}'_i) $ for ${\bf m}' \equiv ({\tt q}'_1,{\tt p}'_1, \dots, {\tt q}'_n,{\tt p}'_n)$   and where $\delta$ denotes the Dirac-delta function in the continuous case and a Kronecker-delta in the discrete case. 
It then follows from Eq.~\eqref{defnpointoperators} that
\begin{eqnarray}
\sum_{{\bf m}\in \Omega} \hat{A}({\bf m}) &=&   \sum_{{\bf m}' \in \Omega} \delta({\bf m}') \hat{W}({\bf m}'),\nonumber\\
&=& \hat{W}({\bf 0}),\nonumber\\
&=& \mathbb{1}.
\end{eqnarray}
Consequently the trace of an arbitrary operator is given by the normalization of the corresponding Wigner representation on the phase-space, 
\[
{\rm Tr}(\hat{O})= \sum_{{\bf m}\in \Omega} W_{\hat{O}}({\bf m}).
\]

The phase-space point operators are Hermitian, and therefore the Wigner representation of any Hermitian operator is real-valued.
They are orthogonal, 
\begin{equation}\label{OrthogPoint}
{\rm Tr}\left( \hat{A}({\bf m})  \hat{A}({\bf m}') \right) \propto \delta({\bf m} - {\bf m}'),
\end{equation}
and form a complete basis for the operator space relative to the Hilbert-Schmidt inner product, that is, for arbitrary $\hat{O}$,
\[
\sum_{{\bf m}\in \Omega} \hat{A}({\bf m}) {\rm Tr}\left(   \hat{A}({\bf m}) \hat{O} \right) = \hat{O}.
\]
It follows from this completeness that for any pair of Hermitian operators $\hat{O}$ and $\hat{O}'$,
\begin{equation}
  \label{eq:HSinnerproductinWigrep}
  \mathrm{Tr}\left( \hat{O} \hat{O}'\right) = \sum_{{\bf m}\in \Omega} W_{\hat{O}}({\bf m}) W_{\hat{O}'} ({\bf m})\,.
\end{equation}


The Wigner representation of a quantum state $\rho$ is the function $W_{\rho}: \Omega \to \mathbb{R}$ defined by
\begin{equation}
W_{\rho}({\bf m}) = {\rm Tr}[\rho \hat{A}({\bf m})],
\end{equation}
where the fact that ${\rm Tr}(\rho)=1$ implies
\begin{equation}
\sum_{{\bf m} \in \Omega} W_{\rho}({\bf m}) = 1.
\end{equation}



A superoperator $\mathcal{E}$ corresponding to the transformation $\rho \mapsto \mathcal{E}(\rho)$ can be modelled in the Wigner representation by a conditional quasiprobability function $W_{\mathcal{E}}({\bf m}'| {\bf m})$ such that 
\[
W_{\rho}({\bf m}) \mapsto \sum_{{\bf m}'\in \Omega} W_{\mathcal{E}}({\bf m}| {\bf m}') W_{\rho}({\bf m}').
\]
Specifically, the function $W_{\mathcal{E}}: \Omega \times \Omega \to \mathbb{R}$ is defined as 
\begin{equation}\label{WignerSuperoperator}
W_{\mathcal{E}}({\bf m}| {\bf m}') = {\rm Tr}\left[ \hat{A}({\bf m}) \mathcal{E}\left(\hat{A}({\bf m}')\right) \right]
\end{equation}
If $\mathcal{E}$ is trace-preserving, then 
\begin{equation}
\sum_{{\bf m}\in \Omega} W_{\mathcal{E}}({\bf m}| {\bf m}') = 1.
\end{equation}

A sharp measurement with outcome set $K$, associated with a projector-valued measure
$\mathcal{O} \equiv \{\hat{\Pi}_{\bf k}: {\bf k} \in K\}$ is represented by a conditional quasi-probability function $W_{\mathcal{O}} : K \times \Omega \to \mathbb{R}$ defined by
\begin{align}
W_{\mathcal{O}}({\bf k}| {\bf m}) &= W_{\hat{\Pi}_{{\bf k}}}({\bf m}),\nonumber\\
&= {\rm Tr}[\hat{\Pi}_{{\bf k}} \hat{A}({\bf m})],
\end{align}
where the fact that $\sum_{{\bf k} \in K} \hat{\Pi}_{\bf k} =\hat{\mathbb{1}}$ implies that
\begin{equation}
\sum_{{\bf k} \in K} W_{\mathcal{O}}({\bf k}| {\bf m}) = 1.
\end{equation}

Finally, we can infer from Eq.~\eqref{eq:HSinnerproductinWigrep} that the probability of obtaining outcome ${\bf k}$ in a measurement of $\{\hat{\Pi}_{\bf k}: {\bf k} \in K\}$ on the state $\rho$ is 
expressed in the Wigner representation as
\begin{equation}\label{kkk}
\mathrm{Tr}\left( \hat{\Pi}_{\bf k} \rho  \right) = \sum_{{\bf m}\in \Omega} W_{\mathcal{O}}({\bf k}| {\bf m}) W_{\rho}({\bf m})  \,.
\end{equation}
Similarly, if a transformation associated with the completely-positive trace-preserving map $\mathcal{E}$ acts between the preparation and the measurement, then the probability of obtaining outcome $k$ is expressed in the Wigner representation as
\begin{equation}\label{lll}
\mathrm{Tr}\left( \hat{\Pi}_{\bf k} \mathcal{E} (\rho)  \right) = \sum_{{\bf m}\in \Omega}   W_{\mathcal{O}}({\bf k}| {\bf m}) \sum_{{\bf m}'\in \Omega} W_{\mathcal{E}}({\bf m}| {\bf m}') W_{\rho}({\bf m}') \,.
\end{equation}

Note that if $W_{\rho}({\bf m})$ is nonnegative, it can be interpreted as a probability distribution on phase-space.  Similarly, if $W_{\mathcal{O}}({\bf k}| {\bf m}) $ and $W_{\mathcal{E}}({\bf m}| {\bf m}')$ are nonnegative, then can be interpreted as conditional probability distributions.
In this case, Eqs.~\eqref{kkk} and \eqref{lll} for the probability of a measurement outcome can be understood as an application of the law of total probability, in analogy with Eqs.~\eqref{opstat1} and \eqref{opstat2}.  This sort of interpretation is indeed possible for the quadrature quantum subtheories and yields precisely the operational predictions of the quadrature epistricted theory.  To show this, 
it remains only to show that the Wigner representation of the preparations, transformations and measurements of the quadrature subtheory are precisely equal to those of the quadrature epistricted theory.


\subsubsection{Wigner representation of the quadrature quantum subtheory}

Our proof of equivalence relies on two of the defining features of the Wigner representation.
First, the fact that the Wigner representation transforms covariantly under the symplectic affine transformations,
Second, the fact that the Wigner representation of the projectors defining the position and momentum observables are the response functions associated to the position and momentum functionals in the classical theory, that is,
\begin{eqnarray}
W_{\hat{\Pi}_{q_i}({\tt q}_i)}({\bf m}) &=& \delta(q_i({\bf m})-{\tt q}_i),\nonumber \\
W_{\hat{\Pi}_{p_j}({\tt p}_j)}({\bf m}) &=& \delta(p_j({\bf m})-{\tt p}_j).
\end{eqnarray}

It follows from these facts that the Wigner representation of the projectors in the quadrature observable $\mathcal{O}_f$ are equal to the response functions associated with the corresponding quadrature functional $f$ in the classical theory,
\begin{eqnarray}\label{quadWig}
W_{\hat{\Pi}_{f}({\tt f})}({\bf m}) 
&=& W_{\hat{\Pi}_{S_f q_1}({\tt f})}({\bf m}),\nonumber\\
&=& W_{\mathcal{V}(S_f)(\hat{\Pi}_{q_1}({\tt f}))}({\bf m}),\nonumber\\
&=& W_{\hat{\Pi}_{q_1}({\tt f})}(S_f^{-1}{\bf m}),\nonumber\\
&=& \delta(q_1(S_f^{-1}{\bf m})-{\tt f}),\nonumber\\
&=& \delta((S_f q_1)({\bf m})-{\tt f}),\nonumber\\
&=& \delta(f({\bf m})-{\tt f})\nonumber\\
&=& \delta({\bf f}^T {\bf m}-{\tt f}).
\end{eqnarray}


As noted previously, the sharp measurements that are included in the quadrature quantum subtheory are those associated to a set of commuting quadrature observables, $\{ \mathcal{O}_{f^{(i)}} \}$  which in turn is associated with a PVM $\mathcal{O}_{V'}\equiv \{ \Pi_{V',{\bf v'}}: {\bf v}' \in V'\}$ where $V' = {\rm span}\{ {\bf f}^{(i)} \}$.  
Given that $\hat{\Pi}_{V}({\bf v}) \equiv \prod_{\{{\bf f}^{(i)}: {\rm span}({\bf f}^{(i)})=V\}} \hat{\Pi}_{f^{(i)}}\left({\bf f}^{(i)T}{\bf v} \right)$ (Eq.~\eqref{pivv}), and using Eq.~\eqref{quadWig}, we conclude that 
\begin{align}\label{Piexp}
W_{\hat{\Pi}_{V'}({\bf v}')}({\bf m})
&= \prod_{\{{\bf f}^{(i)}: {\rm span}({\bf f}^{(i)})=V\}} \delta({\bf f}^{(i)T}{\bf m}- {\bf f}^{(i)T}{\bf v}).
\end{align}
Recalling Eq.~\eqref{responsefns}, we see that the Wigner representation of the projector valued measure associated with $(V',{\bf v}')$ is the set of response functions associated with $(V',{\bf v}')$ in the quadrature epistricted theory, that is, 
\[
W_{\mathcal{O}_{V'}}({\bf v}'|{\bf m})= \xi_{V'}({\bf v}'|{\bf m}).
\]

The Wigner representation of the quadrature state associated with $(V,{\bf v})$ is
\begin{eqnarray}\label{mmm}
W_{\rho_{V,{\bf v}}}({\bf m}) &=&
{\rm Tr}\left(\rho_{V,{\bf v}} \hat{A}({\bf m}) \right) \nonumber\\
&=&  \frac{1}{\mathcal{N}_V } \prod_{{\bf f}^{(i)} : {\rm span}\{{\bf f}^{(i)} \} =V}  \delta ({\bf f}^{(i)T} {\bf m} - {\bf f}^{(i)T}{\bf v})
\end{eqnarray}
where we have used Eqs.~\eqref{qdelt} and \eqref{Piexp}.  Recalling Eqs.~\eqref{epiprod} and \eqref{defndelta}, we conclude that the Wigner representation of the quadrature state associated with $(V,{\bf v})$ is the epistemic state associated with $(V,{\bf v})$ in the quadrature epistricted theory, that is, 
\[
W_{\rho_{V,{\bf v}}}({\bf m}) = \mu_{V,{\bf v}}({\bf m}).
\]

The Wigner representation of the Clifford superoperator $\mathcal{U}(S,{\bf a})$ is 
\begin{eqnarray}\label{proofequivtransf}
W_{\mathcal{U}(S,{\bf a})}({\bf m}|{\bf m}') &=&
 {\rm Tr}\left[ \hat{A}({\bf m})\mathcal{U}(S,{\bf a})\left(\hat{A}({\bf m}')\right) \right]
 \nonumber\\
&=&  {\rm Tr}\left[ \hat{A}({\bf m})\hat{A}( S{\bf m}' + {\bf a}) \right]
 \nonumber\\
&=& 
 \delta( {\bf m} -(S{\bf m}'+{\bf a})).
\end{eqnarray}
Here, the first equality follows from the form of the Wigner representation of superoperators, Eq.~\eqref{WignerSuperoperator}.  The second equality follows from the fact that the phase-space point operators transform covariantly  under the action of the Clifford superoperators, Eq.~\eqref{pointcovariance}.
The third equality in Eq.~\eqref{proofequivtransf} follows from the orthogonality of the phase-space point operators, Eq.~\eqref{OrthogPoint}.

Recalling Eq.~\eqref{Cliffordcond}, we see that this is precisely the transition probability associated with the symplectic affine transformation, $\Gamma_{S,{\bf a}}({\bf m}|{\bf m}')$, in the quadrature epistricted theory,
\begin{eqnarray}\label{proofequivtransf}
W_{\mathcal{U}(S,{\bf a})}({\bf m}|{\bf m}')&=&
 \Gamma_{S,{\bf a}}({\bf m}|{\bf m}').
\end{eqnarray}

This concludes the proof of equivalence.

\color{black}

\subsection{Inequivalence for bits/qubits}

In the case where $d=2$, the only even prime, the situation is more complicated.  We have shown that in \emph{both} the quadrature epistricted theory of bits and in the quadrature subtheory of qubits, we have: (i) the set of possible operational states is isomorphic to the set of pairs $(V,{\bf v})$ where $V$ is an isotropic subspace of the phase-space $\Omega = (\mathbb{Z}_2)^{2n}$ and ${\bf v}\in V$; (ii) the set of possible sharp measurements is isomorphic to the set of isotropic subspaces $V'$ (with the different outcomes associated to the elements ${\bf v}' \in V'$); (iii) the set of possible reversible transformations is isomorphic to the elements $(S,{\bf a})$ of the symplectic affine group acting on $\Omega = (\mathbb{Z}_2)^{2n}$.  It follows that the operational states, measurements and transformations of one theory are respectively isomorphic to those of the other.  The valid \emph{unsharp} measurements and \emph{irreversible} transformations are defined in terms of the sharp and reversible ones respectively, and they are defined \emph{in the same way} in the quadrature epistricted theory and the quadrature quantum subtheory.  It follows that we also have the unsharp measurements and irreversible transformations of one theory isomorphic to those of the other. 

Despite this strong structural similiarly, the two theories nonetheless make different predictions. The particular algorithm that takes as input a triple of preparation, measurement and transformation and yields as output a probability distribution over measurement outcomes, is not equivalent in the two theories.  More precisely,
\[
{\rm Tr}\left(\hat{\Pi}_{V'}({\bf v}') \mathcal{U}_{S,{\bf a}} (\rho_{V,{\bf v}})\right) \nonumber\\
\ne  \sum_{{\bf m} \in \Omega} \xi_{V'}({\bf v}' | {\bf m}) \sum_{{\bf m}' \in \Omega} \Gamma_{S,{\bf a}}({\bf m}|{\bf m}')  \mu_{V,{\bf v}}({\bf m}').
\]

Gross's Wigner representation for discrete systems only works for systems of dimension $d$ for $d$ a power of an odd prime.  It therefore does not work for $d=2$.  Nonetheless, a Wigner representation can be constructed for the quadrature subtheory of qubits. 
One can define it in terms of tensor products of the phase-space point operators for a single qubit as proposed by Gibbons, Hoffman and Wootters~\cite{gibbons2004discrete}.  In this representation, we have
\[
{\rm Tr}\left(\hat{\Pi}_{V'}({\bf v}')\; \mathcal{U}({S,{\bf a}}) \left(\rho_{V,{\bf v}}\right)\right) \nonumber\\
=  \sum_{{\bf m} \in \Omega} W_{\mathcal{O}_{V'}}({\bf v}' | {\bf m}) \sum_{{\bf m}' \in \Omega} W_{\mathcal{U}(S,{\bf a})}({\bf m}|{\bf m}')  W_{\rho_{V,{\bf v}}}({\bf m}').
\]

So why can't we identify the Wigner representations with the corresponding objects in the epistricted theory, just as we did for $d$ an odd prime and in the continuous case?  The problem is that in the qubit case, the Wigner functions representing quadrature states sometimes go negative.  It follows that these cannot be interpreted as probability distributions over the phase space.  Similarly, the Wigner representations  of quadrature observables and Clifford superoperators also sometimes go negative and hence cannot always be interpreted as conditional probability distributions.

It is also straightforward to prove that no alternative definition of the Wigner representation can achieve positivity.   First, we make use of a fact shown in Ref.~\cite{spekkens2008negativity}, that if a set of preparations and measurements supports a proof of contextuality in the sense of Ref.~\cite{spekkens2005contextuality}, then \emph{all} quasiprobability representations must necessarily involve negativity.  It then suffices to note that the quadrature subtheory is contextual.   There are many ways of seeing this.  For instance, Mermin's magic square proof of contextuality using two qubits~\cite{mermin1993hidden} uses only the resources of the stabilizer theory of qubits.   The same is true of the Greenberger-Horne-Zeilinger proof of nonlocality using three qubits~\cite{greenberger1989going}, which is also a proof of contextuality. 

The quadrature subtheory of qubits simply makes different operational predictions than the quadrature epistricted theory of bits.  It admits of contextuality and nonlocality proofs while the quadrature epistricted theory is local and noncontextual by construction.\footnote{It seems that the quadrature epistricted theory of bits is about as close as one can get to the stabilizer theory for qubits while still being local and noncontextual.}   

By contrast, the quadrature subtheory of odd-prime qudits and the quadrature subtheory of mechanics make precisely the same predictions as the corresponding epistricted theories.  They are consequently devoid of any contextuality or nonlocality because they admit of hidden variable models that are both local and noncontextual---the quadrature epistricted theory \emph{is} the hidden variable model.  Other differences between the two theories are discussed in Ref.~\cite{spekkens2007evidence}.

The conceptual significance of the difference between the quadrature subtheory of qubits and the epistricted theory of bits remains a puzzle, despite various formalizations of the difference~\cite{pusey2012stabilizer,coecke2011phase}.
This puzzle is perhaps the most interesting product of these investigations.  

It shows in particular that whatever conceptual innovation over the classical worldview is required to achieve the phenomenology of contextuality and nonlocality, it must be possible to make sense of this innovation even in the thin air of the quadrature subtheory of qubits.  This is an advantage because the latter theory uses a more meagre palette of concepts than full quantum theory.
For instance, we can conclude that it must be possible to describe the innovation of quantum over classical in terms of possibilistic inferences rather than probabilistic inferences.

\begin{acknowledgements}
I acknowledge Stephen Bartlett and Terry Rudolph for discussions on the quadrature subtheory of quantum mechanics, Jonathan Barrett for suggesting to define the Poisson bracket in the discrete case in terms of finite differences, and Giulio Chiribella and Joel Wallman for comments on a draft of this article.  Much of the work presented here summarizes unpublished results obtained in collaboration with Olaf Schreiber.  Research at Perimeter Institute is supported by the Government of Canada through Industry Canada and by the Province of Ontario through the Ministry of Research and Innovation.
\end{acknowledgements}

\begin{appendix}

\section{Quadrature quantum subtheories and the Stabilizer formalism}


In quantum information theory, there has been a great deal of work on a particular quantum subtheory for discrete systems of prime dimension (qubits and qutrits in particular) which is known as the \emph{stabilizer formalism}~\cite{gottesman1998heisenberg,gross2006hudson}.

A stabilizer state is defined as a joint eigenstate of a set of commuting Weyl operators.
By Eq.~\eqref{productWeyl}, two Weyl operators commute if and only if the corresponding phase-space displacement vectors have vanishing symplectic inner product, 
\[
[ \hat{W}({\bf a}) , \hat{W}({\bf a}') ]=0 \;\textrm{  if and only if }\;   \langle {\bf a}, {\bf a}'\rangle=0.
\]
Consequently, the sets of commuting Weyl operators, and therefore the stabilizer states, are parametrized by the isotropic subspaces of $\Omega$.  Specifically, for each isotropic subspace $M$ of $\Omega$ and each vector ${\bf v}\in JM\equiv \{ J{\bf u} : {\bf u} \in M \}$, we can define a stabilizer state $\rho^{(\rm stab)}_{M,{\bf v}}$ as the projector onto the joint eigenspace of $\{ \hat{W}({\bf a}): {\bf a} \in M\}$ where $\hat{W}({\bf a})$ has eigenvalue $\chi(\langle {\bf v}, {\bf a} \rangle)$.

 We will show here that the set of stabilizer states is precisely equivalent to the set of quadrature states.

To describe the connection, it is convenient to introduce some additional notions from symplectic geometry.
The \emph{symplectic complement} of a subspace $V$, which we will denote as $V^C$, is the set of vectors that have vanishing symplectic inner product with every vector in $V$,
\[
V^{C}\equiv \{ {\bf m}' \in \Omega : {\bf m'}^T J {\bf m} =0, \; \forall {\bf m} \in V \},
\]
where $J$ is the symplectic form, defined in Eq.~\eqref{eq:SymplecticForm}. 
This is not equivalent to the Euclidean complement of a subspace $V$, which is the set of vectors that have vanishing Euclidean inner product with every vector in $V$,
\[
V^{\perp}\equiv \{ {\bf m}' \in \Omega : {\bf m'}^T {\bf m} =0, \; \forall {\bf m} \in V \},
\]
The composition of the two complements will be relevant in what follows.  It turns out that the latter is related to $V$ by an isomorphism; it is simply the image of $V$ under left-multiplication by the symplectic form $J$,  
\[
(V^{\perp})^{C}\equiv JV = \{ J{\bf u} : {\bf u} \in V \}.
\]
Note that if $V$ is isotropic, then $(V^{\perp})^{C}$ is as well. 



\begin{proposition}\label{equivstabquad}
Consider the quadrature state $\rho_{V,{\bf v}}$, with $V$ an isotropic subspace of $\Omega$ and ${\bf v}\in V$ a valuation vector, which is the joint eigenstate of the commuting set of quadrature observables $\{ \mathcal{O}_{f} : {\bf f}\in V \}$, where the eigenvalue of $\mathcal{O}_f$ is $f({\bf v})$.
This is equivalent to the stabilizer state $\rho^{({\rm stab})}_{M,{\bf v}}$, which is the joint eigenstate of the commuting set of Weyl operators $\{ \hat{W}({\bf a}) : {\bf a} \in M\}$ where $M \equiv (V^{\perp})^C$ is the isotropic subspace that is the symplectic complement of the Euclidean complement of $V$, and where the eigenvalue of $\hat{W}({\bf a})$ is $\chi(\langle {\bf v},{\bf a} \rangle)$.
\end{proposition}

\begin{proof}
Consider first a single degree of freedom.  Every quadrature observable $\mathcal{O}_{f}$ can be expressed in terms of the position observable $\mathcal{O}_q$ as follows: if $S_f$ is the symplectic matrix such that ${\bf f} = S_f {\bf q}$, then $\mathcal{O}_{f} = \hat{V}(S_f) \mathcal{O}_{q} \hat{V}(S_f)^{\dag}$.  Now note that the position basis can equally well be characterized as the eigenstates of the boost operators.  Specifically, $\hat{B}({\tt p}) |{\tt q}\rangle_q = \chi({\tt q p})  |{\tt q}\rangle_q$, that is, an element $|{\tt q}\rangle_q$ of the position basis is an eigenstate of the set of operators $\{ \hat{B}({\tt p}): {\tt p}\in \mathbb{R}/\mathbb{Z}_d \}$ where the eigenvalue of $\hat{B}({\tt p})$  is $\chi({\tt qp})$.  The element $|{\tt f}\rangle_f$ of the basis associated to the quadrature operator $\mathcal{O}_f$ is defined as $|{\tt f}\rangle_f \equiv \hat{V}(S_f)|{\tt f}\rangle_q$ and consequently can be characterized as an eigenstate of the set of operators $\{ \hat{V}(S_f) \hat{B}({\tt g}) \hat{V}(S_f)^{\dag}: {\tt g}\in \mathbb{R}/\mathbb{Z}_d \}$ where the eigenvalue of $\hat{V}(S_f) \hat{B}({\tt g}) \hat{V}(S_f)^{\dag}$  is $\chi({\tt fg})$.  This can be stated equivalently as follows: the element $|{\tt f}\rangle_f$ of the basis associated to the quadrature operator $\mathcal{O}_f$ is the eigenstate of the set of Weyl operators $\{ \hat{W}({\bf a}): {\bf a}\in {\rm span}(S_f{\bf p}) \}$ where the eigenvalue of $\hat{W}({\bf a})$  is $\chi( {\tt f}\langle {\bf f},{\bf a}\rangle)$. Noting that 
\begin{eqnarray}
{\rm span}(S_f{\bf p})&={\rm span}(S_f J{\bf q}),\nonumber\\
&= {\rm span}(J S_f {\bf q}),\nonumber\\
&= {\rm span}(J {\bf f}),
\end{eqnarray}
 we can just as well characterize $\hat{\Pi}_f({\tt f})$ as the projector onto the joint eigenspace of the Weyl operators $\{ \hat{W}({\bf a}): {\bf a}\in {\rm span}(J {\bf f}) \}$.


Now consider $n$ degrees of freedom. The quadrature state associated with $(V,{\bf v})$ has the form
\begin{equation}\label{pivv}
\rho_{V,{\bf v}} 
= \frac{1}{\mathcal{N}} \prod_{{\bf f}^{(i)}: {\rm span}({\bf f}^{(i)})=V} \hat{\Pi}_{f^{(i)}}\left(f^{(i)}({\bf v})\right).
\end{equation}
By an argument similar to that used for a single degree of freedom, this
is an eigenstate of the Weyl operators $\{ \hat{W}({\bf a}): {\bf a}\in {\rm span}(J {\bf f}^{(i)}) \}$ 
where the eigenvalue of $\hat{W}({\bf a})$  is $\chi( \langle {\bf v}, {\bf a} \rangle)$.
Noting that ${\rm span}(J {\bf f}^{(i)}) = J V = (V^{\perp})^C$,  we have our desired isomorphism.  
\end{proof}

The stabilizer formalism allows all and only the Clifford superoperators as reversible transformations. The sharp measurements that are included in the stabilizer formalism are the ones associated with PVMs corresponding to the joint eigenspaces of a set of commuting Weyl operators, which, by proposition \ref{equivstabquad}, are precisely those corresponding to the joint eigenspaces of a set of commuting quadrature observables.  It follows that the stabilizer formalism coincides precisely with the quadrature subtheory.  

Gross has argued that the discrete analogue of the Gaussian quantum subtheory for continuous variable systems is the stabilizer formalism~\cite{gross2006hudson}.  Our results show that the connection between the discrete and continuous variable cases is a bit more subtle than this.  In the continuous variable case, there is a distinction between the Gaussian subtheory and the quadrature subtheory, with the latter being contained within the former.  In the discrete case, there is no distinction, so the stabilizer formalism can be usefully viewed as either the discrete analogue of the Gaussian subtheory or as the discrete analogue of the quadrature subtheory.  While Gross's work showed that the stabilizer formalism for discrete systems could be defined similarly to how one defines Gaussian quantum mechanics, our work has shown that it can also be defined in the same way that one defines quadrature quantum mechanics.

To our knowledge, quadrature quantum mechanics has not previously received much attention.  However, given that it is a natural continuous variable analogue of the stabilizer formalism for discrete systems, it may provide an interesting paradigm for exploring quantum information processing with continuous variable systems.

\end{appendix}

\bibliographystyle{unsrt}
\bibliography{EpistrictedTheoriesReview.bib}

\end{document}